\documentclass{article}
\usepackage{fullpage}
\usepackage{microtype}
\usepackage{graphicx}
\usepackage{caption}
\usepackage{subcaption}
\usepackage{booktabs} 
\usepackage{multirow}
\usepackage[colorlinks=true, allcolors=blue]{hyperref}
\usepackage[square,sort]{natbib}
\usepackage{appendix}
\usepackage{tikz}
\usetikzlibrary{shapes.geometric, arrows}
\usepackage{algorithm}
\usepackage{algorithmic}

\usepackage{amsmath}
\usepackage{amssymb}
\usepackage{amsthm}
\usepackage{xr}

\usepackage{bbm}
\usepackage{xspace}
\newcommand{\apt}{i-FWER test\xspace}
\newcommand{\one}{\mathbbm{1}\xspace}
\DeclareMathOperator*{\argmax}{argmax}
\DeclareMathOperator*{\argmin}{argmin}
\newtheorem{theorem}{Theorem}
\newtheorem{lemma}{Lemma}
\newtheorem{corollary}{Corollary}
\newtheorem*{setting}{Setting}
\newtheorem{remark}{Remark}

\begin{document}

\begin{center}

{\bf{\LARGE{Familywise Error Rate Control by \\ Interactive Unmasking}}}

\vspace*{.2in}

{\large{
\begin{tabular}{cccc}
Boyan Duan &  Aaditya Ramdas & Larry Wasserman \\
\end{tabular}
\texttt{\{boyand,aramdas,larry\}@stat.cmu.edu}
}}

\vspace*{.2in}

\begin{tabular}{c}
Department of Statistics and Data Science,\\
Carnegie Mellon University,  Pittsburgh, PA  15213
\end{tabular}

\vspace*{.2in}

\today

\vspace*{.2in}

\begin{abstract}
We propose a method for multiple hypothesis testing with familywise error rate (FWER) control, called the \apt. Most testing methods are predefined algorithms that do not allow modifications after observing the data. However, in practice, analysts tend to choose a promising algorithm after observing the data; unfortunately, this violates the validity of the conclusion. The \apt allows much flexibility: a human (or a computer program acting on the human's behalf) may adaptively guide the algorithm in a data-dependent manner. We prove that our test controls FWER if the analysts adhere to a particular protocol of “masking” and “unmasking”. We demonstrate via numerical experiments the power of our test under structured non-nulls, and then explore new forms of masking.
\end{abstract}
\end{center}

\section{Introduction} \label{sec::intro}

Hypothesis testing is a critical instrument in scientific research to quantify the significance of a discovery. For example, suppose an observation $Z \in \mathbb{R}$ follows a Gaussian distribution with mean $\mu$ and unit variance. We wish to distinguish between the following null and alternative hypotheses regarding the mean value:
\begin{align} \label{eq:normal_hp}
    H_0: \mu \leq 0 \quad \text{versus} \quad H_1: \mu > 0.
\end{align}
A test decides whether to reject the null hypothesis, usually by calculating a \emph{$p$-value}: the probability of observing an outcome at least as extreme as the observed data under the null hypothesis. In the above example, the $p$-value is $P = 1 - \Phi(Z)$, where $\Phi$ is the cumulative distribution function (CDF) of a standard Gaussian. When the true mean~$\mu$ is exactly zero, the $p$-value is uniformly distributed; when $\mu < 0$, it has nondecreasing density. A low $p$-value suggests evidence to reject the null hypothesis.

Recent work on testing focuses on a large number of hypotheses, referred to as \emph{multiple testing}, driven by various applications in Genome-wide Association Studies, medicine, brain imaging, etc. (see \citep{farcomeni2008review, goeman2014multiple} and references therein). In such a setup, we are given $n$ null hypotheses $\{H_i\}_{i = 1}^n$ and their $p$-values $P_1, \ldots, P_n$. A multiple testing method examines the $p$-values, possibly together with some side/prior information, and decides whether to reject each hypothesis (i.e., infers which ones are the non-nulls). Let $\mathcal{H}_0$ be the set of hypotheses that are truly null and $\mathcal{R}$ be the set of rejected hypotheses, then $V = |\mathcal{H}_0 \cap \mathcal{R}|$ is the number of erroneous rejections. This paper considers a classical error metric, \textit{familywise error rate}:
\[
\text{FWER} := \mathbb{P}(V \geq 1),
\]
which is the probability of making any false rejection. Given a fixed level $\alpha \in (0,1)$, a good test should have valid error control that FWER $\leq \alpha$, and high \emph{power}, defined as the expected proportion of rejected non-nulls:
\[
\text{power} := \mathbb{E}\left(\frac{|\mathcal{R} \backslash \mathcal{H}_0|}{|[n] \backslash \mathcal{H}_0|}\right),
\]
where $[n] := \{1,\ldots, n\}$ denotes the set of all hypotheses. 

Most methods with FWER control follow a prespecified algorithm (see, for instance, \citep{holm1979simple, hochberg1988sharper, bretz2009graphical, goeman2011multiple, tamhane2018advances} and references therein). However, in practice, analysts tend to try out several algorithms or parameters on the same dataset until results are ``satisfying''. When a second group repeats the same experiments, the outcomes are often not as good. This problem in reproducibility comes from the bias in selecting the analysis tool: researchers choose a promising method after observing the data, which violates the validity of error control. Nonetheless, data would greatly help us understand the problem and choose an appropriate method if it were allowed. This motivates us to propose an interactive method called the \emph{\apt}, that (a)~can use observed data in the design of testing algorithm, and (b)~is a multi-step procedure such that a human can monitor the performance of the current algorithm and is allowed to adjust it at any step interactively; and still controls FWER.

\begin{figure*}[t]
    \centering
\tikzstyle{op} = [rectangle, dashed, minimum width=1.5cm, minimum height=0.6cm, text centered, draw=black, fill=none]
\tikzstyle{rec} = [rectangle, minimum width=1.5cm, minimum height=0.6cm, text centered, draw=black, fill=none]
\tikzstyle{rrec} = [rectangle, rounded corners, minimum width=1.5cm, minimum height=0.6cm, text centered, draw=purple]
\tikzstyle{arrow} = [thick,->,>=stealth]

\begin{tikzpicture}[node distance=2cm]
\label{flow:exclude}

\node (pval)[rec] {$p$-values $\{P_i\}$};
\node (gp) [rec, right of=pval, xshift = 1.5cm, yshift = 1.5cm]{$\{g(P_i)\}$};
\node (hp) [rec, right of=pval, xshift = 1.5cm, yshift = -1.5cm]{$\{h(P_i)\}$};
\node (side) [op, above of= gp, xshift = 2.3cm, yshift = -1cm]{Prior/side information, covariates $\{x_i\}$};
\node (rejset) [rec, right of=gp, xshift = 4cm]{Rejection set $\mathcal{R}_t$};
\node (error) [rec, right of=hp, xshift = 4cm]{Estimate $\widehat{\text{FWER}}_t$};
\node (reject) [right of=error, xshift = 2.5cm]{Report $\mathcal{R}_t$};

\node (select) [above of=rejset, xshift = 0.85cm, yshift = -0.2cm,  color = cyan]{Selection};
\node (control) [below of=error, xshift = 0.6cm, yshift = 0.9cm, color = red]{Error control};

\draw [arrow] (pval) -- node[rrec, anchor= south, pos = 0.5, yshift = 0.1cm, color = purple] {Masking} (2.5,0) -- (gp);
\draw [arrow] (2.5,0) -- (hp);
\draw [arrow, ultra thick] (gp) -- node[anchor= south, pos = 0.5, xshift = 1cm] {Shrink} (rejset);
\draw [arrow] (rejset) -- (error);
\draw [arrow] ([yshift=-0.2em]hp.east) -- ([yshift=-0.2em]error.west);
\draw [arrow] (error) -- node[anchor= south, pos = 0.6] {If $\widehat{\text{FWER}}_t \leq \alpha$} (reject);
\draw [arrow, dashed] (side) -- (5,1.6);
\draw [arrow, dashed] ([yshift=0.2em]error.west) --node[anchor= south, pos = 0.5] {If $\widehat{\text{FWER}}_t > \alpha$} ([yshift=0.2em]hp.east);
\draw [arrow, dashed] (hp) -- node[anchor= west, pos = 0.5] {Unmasking} (5,1.4);

\draw[rounded corners=15pt, color=cyan, thick]
  (2.5,0.8) rectangle ++(8.6,2.2);
\draw[rounded corners=15pt, color=red, thick]
  (2.5,-0.5) rectangle ++(8.6,-1.8);
\end{tikzpicture}
    \caption{A schematic of the \apt. All $p$-values are initially `masked': all $\{g(P_i)\}$ are revealed to the analyst/algorithm, while all $\{h(P_i)\}$ remain hidden, and the initial rejection set is $\mathcal{R}_0=[n]$. If $\widehat{\text{FWER}}_t > \alpha$, the analyst chooses a $p$-value to `unmask' (observe the masked $h(P)$-value), effectively removing it from the proposed rejection set $\mathcal R_t$; importantly, using any available side information and/or covariates and/or working model, the analyst can shrink $\mathcal R_t$ in any manner. This process continues until $\widehat{\text{FWER}}_t \leq \alpha$ (or $\mathcal R_t = \emptyset$).}
    \label{fig:flow_apt}
\end{figure*}
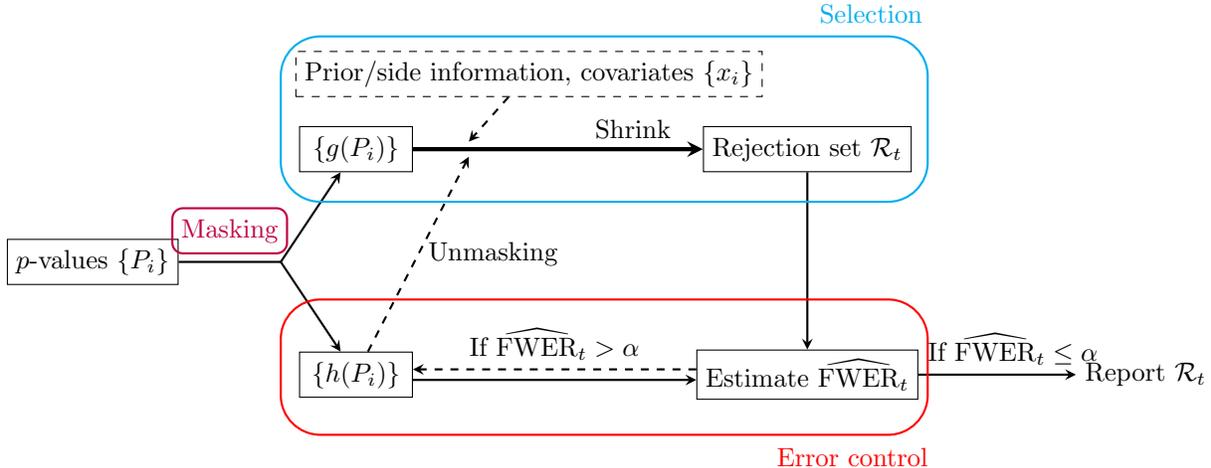

The word ``interactive'' is used in many contexts in machine learning and statistics. Specifically, multi-armed bandits, active learning, online learning, reinforcement learning, differential privacy, adaptive data analysis, and post-selection inference all involve some interaction. Each of these paradigms has a different goal, a different model of interaction, and different mathematical tools to enable and overcome the statistical dependencies created by data-dependent interaction. The type of interaction proposed in this paper is different from the above. Here, the goal is to control FWER in multiple testing. The model of interaction involves ``masking'' of $p$-values followed by progressive unmasking (details in the next paragraph). The technical tools used are (a) for $p$-values of the true nulls (\emph{null $p$-values}), the masked and revealed information are independent, (b) an empirical upper bound on the FWER that can be continually updated using the revealed information. 

The key idea that permits interaction while ensuring FWER control is ``masking and unmasking'', proposed by \citet{lei2018adapt, lei2017star}.  In our method, it has three main steps and alternates between the last two (Figure~\ref{fig:flow_apt}):
\begin{enumerate}
   \item \textbf{Masking}. Given a parameter $p_* \in (0,1)$, each $p$-value $P_i$ is decomposed into two parts by functions $h:[0,1] \to \{-1,1\}$ and $g:[0,1] \to (0,p_*)$:
   \begin{align}
   \label{eq:decomp_varyp}
    h(P_i; p_*) =~& 2\cdot\one\{P_i < p_*\} - 1;\nonumber {}\\
    \text{and } g(P_i; p_*) =~& \min\left\{P_i, \frac{p_*}{1-p_*}(1-P_i) \right\},
   \end{align}
   where $g(P_i)$, the \textit{masked $p$-value}, is used to interactively adjust the algorithm, and $h(P_i)$, the \textit{revealed missing bit}, is used for error control. Note that $h(P_i)$ and $g(P_i)$ are independent if $H_i$ is null ($P_i$ is uniformly distributed); this fact permits interaction with an analyst without any risk of violating FWER control.
   
   \begin{figure}[t]
    \centering
    \hspace{3cm}
    \begin{subfigure}[t]{0.23\textwidth}
        \centering
        \includegraphics[trim=0cm 0.5cm 1cm 2cm, clip, width=\linewidth]{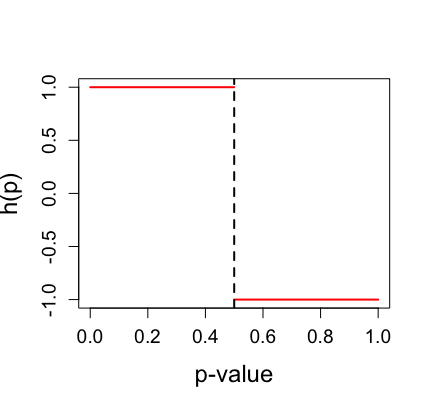}
    \end{subfigure}
    \hfill
    \begin{subfigure}[t]{0.23\textwidth}
        \centering
        \includegraphics[trim=0cm 0.5cm 1cm 2cm, clip, width=\linewidth]{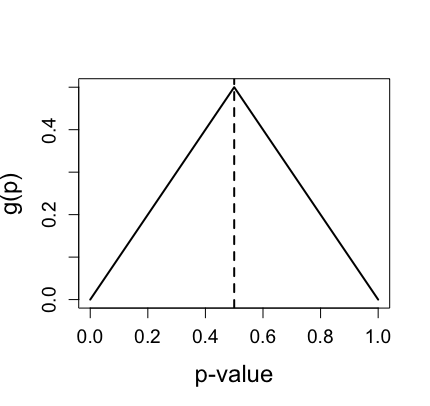}
    \end{subfigure}
    \hspace{3cm}
    \caption{Functions for masking~\eqref{eq:decomp_varyp}: missing bits~$h$ (left) and masked $p$-values~$g$ (right) when $p_* = 0.5$. For uniform $p$-values, $g(P)$ and $h(P)$ are independent.}
    \label{fig:hg}
\end{figure}

   \item \textbf{Selection}. Consider a set of candidate hypotheses to be rejected (rejection set), denoted as $\mathcal{R}_t$ for iteration~$t$. We start with all the hypotheses included, $\mathcal{R}_0 = [n]$. At each iteration, the analyst excludes possible nulls from the previous $\mathcal{R}_{t-1}$, using all the available information (masked $p$-values, progressively unmasked $h(P_i)$ from step 3 and possible prior information). Note that our method does not automatically use prior information and masked $p$-values. The analyst is free to use any black-box prediction algorithm or Bayesian working model that uses the available information, and orders the hypotheses possibly using an estimated likelihood of being non-null. This step is where a human is allowed to incorporate their subjective choices.
   
   \item \textbf{Error control (and unmasking)}. The FWER is estimated using $h(P_i)$. If the estimation $\widehat{\text{FWER}}_t > \alpha$, the analyst goes back to step 2, provided with additional information: unmasked (reveal) $h(P_i)$ of the excluded hypotheses, which improves her understanding of the data and guides her choices in the selection step.
\end{enumerate}

The rest of the paper is organized as follows. In Section~\ref{sec::apt}, we describe the \apt in detail. In Section~\ref{sec::sim}, we implement the interactive test under a clustered non-null structure. In Section~\ref{sec::mask}, we propose two alternative ways of masking $p$-values and explore their advantages.

\section{An interactive test with FWER control} \label{sec::apt}
Interaction shows its power mostly when there is prior knowledge. We first introduce the \emph{side information}, which is available before the test in the form of covariates $x_i$ as an arbitrary vector (mix of binary, real-valued, categorical, etc.) for each hypothesis $i$. For example, if the hypotheses are arranged in a rectangular grid (such as when processing an image), then $x_i$ could be the coordinate of hypothesis $i$ on the grid. Side information can help the analyst to exclude possible nulls, for example, when the non-nulls are believed to form a cluster on the grid by some domain knowledge. Here, we state the algorithm and error control with the side information treated as fixed values, but side information can be random variables, like the bodyweight of patients when testing whether each patient reacts to a certain medication. Our test also works for random side information $X_i$ by considering the conditional behavior of $p$-values given~$X_i$. 

The \apt proceeds as progressively shrinking a candidate rejection set $\mathcal{R}_t$ at step $t$, 
\[
[n] = \mathcal{R}_0 \supseteq \mathcal{R}_1 \supseteq \ldots \supseteq \mathcal{R}_n = \emptyset,
\]
where recall $[n]$ denotes the set of all the hypotheses. We assume without loss of generality that one hypothesis is excluded in each step. Denote the hypothesis excluded at step $t$ as $i_t^*$. The choice of $i_t^*$ can use the information available to the analyst before step $t$, formally defined as a filtration (sequence of nested $\sigma$-fields) \footnote{ This filtration denotes the information used for choosing~$i_t^*$. The filtration with respect to which the stopping time in Algorithm~\ref{alg:apt} is measurable includes the scale of $R_t^-$:
${\mathcal{G}_{t-1} 
    := \sigma\Big(\mathcal{F}_{t-1}, |i\in R_t: h(P_i) = -1|\Big)}$.}:
\begin{align} \label{eq:sigma}
    \mathcal{F}_{t-1} 
    := \sigma\Big(\{x_i, g(P_i)\}_{i = 1}^n, \{P_i\}_{i \notin \mathcal{R}_{t-1}} \Big),
\end{align}
where we unmask the $p$-values for the hypotheses that are excluded from the rejection set $\mathcal{R}_{t-1}$.

To control FWER, the number of false discoveries $V$ is estimated using only the binary missing bits $h(P_i)$. The idea is to partition the candidate rejection set $\mathcal{R}_t$ into $\mathcal{R}_t^+$ and $\mathcal{R}_t^-$ by the value of $h(P_i)$: 
\begin{align*}
    \mathcal{R}_t^+ :=~& \{i\in \mathcal{R}_t: h(P_i) = 1\} \equiv  \{i\in \mathcal{R}_t: P_i < p_*\},{}\\
    \mathcal{R}_t^- :=~& \{i\in \mathcal{R}_t: h(P_i) = -1\} \equiv \{i\in \mathcal{R}_t: P_i \geq p_*\};
\end{align*}
recall that $p_*$ is the prespecified parameter for masking~\eqref{eq:decomp_varyp}. Instead of rejecting every hypothesis in $\mathcal{R}_t$, note that the test only rejects the ones in $\mathcal{R}_t^+$, whose $p$-values are smaller than~$p_*$ in $\mathcal{R}_t$. Thus, the number of false rejection $V$ is ${|\mathcal{H}_0 \cap \mathcal{R}_t^+|}$ and we want to control FWER, ${\mathbb{P}(V \geq 1)}$. The distribution of ${|\mathcal{H}_0 \cap \mathcal{R}_t^+|}$ can be estimated by $|\mathcal{H}_0 \cap \mathcal{R}_t^-|$ using the fact that $h(P_i)$ is a (biased) coin flip. But $\mathcal{H}_0$ (the set of true nulls) is unknown, so we use $|\mathcal{R}_t^-|$ to upper bound $|\mathcal{H}_0 \cap \mathcal{R}_t^-|$, and propose an estimator of FWER:
\begin{align} \label{eq:fwer_hat}
    \widehat{\text{FWER}}_t = 1 - (1 - p_*)^{|\mathcal{R}_t^-| + 1}.
\end{align}
Overall, the \apt shrinks $\mathcal{R}_t$ until $\widehat{\text{FWER}}_t \leq \alpha$ and rejects only the hypotheses in $\mathcal{R}_t^+$ (Algorithm~\ref{alg:apt}). 

\begin{algorithm}[t]
   \caption{The \apt}
   \label{alg:apt}
\begin{algorithmic}
   \STATE {\bfseries Input:} Side information and $p$-values $\{x_i, P_i\}_{i = 1}^n$, target FWER level~$\alpha$, and parameter~$p_*$;
   \STATE {\bfseries Procedure:} 
   \STATE Initialize $\mathcal{R}_0 = [n]$;
   \FOR{$t=1$ {\bfseries to} $n$}
   \STATE 1.~Pick any $i_t^* \in \mathcal{R}_{t-1}$, using $\{x_i, g(P_i)\}_{i=1}^n$ and progressively unmasked $\{h(P_i)\}_{i \notin \mathcal{R}_{t-1}}$;
   \STATE 2.~Exclude $i_t^*$ and update ${\mathcal{R}_t = \mathcal{R}_{t-1}\backslash \{i_t^*\}}$; 
   \IF{$\widehat{\text{FWER}}_t \equiv 1 - (1 - p_*)^{|\mathcal{R}_t^-| + 1} \leq \alpha$}
   \STATE Reject $\{H_i: i \in \mathcal{R}_t, h(P_i) = 1\}$ and exit;
   \ENDIF
   \ENDFOR
\end{algorithmic}
\end{algorithm}

\begin{remark}
\label{rmk:p*}
The parameter $p_*$ should be chosen in $(0,\alpha]$, because otherwise $\widehat{\text{FWER}}_t$ is always larger than $\alpha$ and no rejection would be made. In principle, because $|\mathcal{R}_t^-|$ only takes integer values, we should $p_*$ such that $\frac{\log (1 - \alpha)}{\log (1 - p_*)}$ is an integer; otherwise, the estimated FWER at the stopping time, $\widehat{\text{FWER}}_\tau$, would be strictly smaller than $\alpha$ rather than equal. Our numerical experiments suggest that the power is relatively robust to the choice of $p_*$. A default choice can be $p_* \approx \alpha/2$ (see detailed discussion in Appendix~\ref{apd::para}).
\end{remark}

\begin{remark} \label{rmk:k-fwer}
The above procedure can be easily extended to control $k$-FWER:
\begin{align} \label{eq:k-fwer}
    k\text{-FWER} := \mathbb{P}(V \geq k),
\end{align}
by estimating $k$-FWER as
\begin{align*} 
    \widehat{k\text{-FWER}}_t = 1 - \sum_{i=0}^{k-1}\binom{|\mathcal{R}_t^-|+i}{i}(1 - p_*)^{|\mathcal{R}_t^-| + 1} p_*^i.
\end{align*}
\end{remark}

The error control of \apt uses an observation that at the stopping time, the number of false rejections is stochastically dominated by a negative binomial distribution. The complete proof is in Appendix~\ref{apd:thm1}.
\begin{theorem}
\label{thm:fwer}
Suppose the null $p$-values are mutually independent and they are independent of the non-nulls, then the \apt controls FWER at level $\alpha$.
\end{theorem}

\begin{remark} \label{rmk:cons}
The null $p$-values need not be exactly uniformly distributed. For example, FWER control also holds when the null $p$-values have nondecreasing densities. Appendix~\ref{apd:mirror-cons} presents the detailed technical condition for the distribution of the null $p$-values. 
\end{remark}

\paragraph{Related work.} 
The \apt mainly combines and generalizes two sets of work: (a)~we use the idea of masking from \citet{lei2018adapt, lei2017star} and extend it to a more stringent error metric, FWER; (b)~we use the method of controlling FWER from \citet{janson2016familywise} by converting a one-step procedure in the context of ``knockoff'' statistics in regression problem to a multi-step (interactive) procedure in our context of $p$-values. 

\citet{lei2018adapt} and \citet{lei2017star} introduce the idea of masking and propose interactive tests that control \emph{false discovery rate} (FDR):
\[
\text{FDR} := \mathbb{E}\left(\frac{V}{|\mathcal{R}| \vee 1}\right),
\]
the expected proportion of false discoveries. It is less stringent than FWER, the probability of making \emph{any} false discovery. Their method uses the special case of masking~\eqref{eq:decomp_varyp} when $p_* = 0.5$, and estimate $V$ by $\sum_{i \in \mathcal{R}_t} \one \{h(P_i) = -1\}$, or equivalently $\sum_{i \in \mathcal{R}_t} \one \{P_i < 0.5\}$. While it provides a good estimation on the \emph{proportion} of false discoveries, the indicator $\one\{P_i < 0.5\}$ has little information on the correctness of \emph{individual} rejections. To see this, suppose there is one rejection, then FWER is the probability of this rejection being false. Even if $h(P_i) = 1$, which indicates the $p$-value is on the smaller side, the tightest upper bound on FWER is as high as $0.5$. Thus, our method uses masking~\eqref{eq:decomp_varyp} with small $p_*$, so that $h(P_i) = 1$, or equivalently $P_i < p_*$, suggests a low chance of false rejection.

In the context of a regression problem to select significant covariates, \citet{janson2016familywise} proposes a one-step method with control on $k$-FWER; recall definition in~\eqref{eq:k-fwer}. The FWER is a special case of $k$-FWER when $k = 1$, and as $k$ grows larger, $k$-FWER is a less stringent error metric. Their method decomposes statistics called ``knockoff'' \citep{barber2015controlling} into the magnitudes for ordering covariates (without interaction) and signs for estimating $k$-FWER, which corresponds to decomposing $p$-values into $g(P_i)$ and $h(P_i)$ when $p_* = 0.5$. However, the decomposition as magnitude and sign restricts the corresponding $p$-value decomposition with a single choice of $p_*$ as $0.5$, making the $k$-FWER control conservative and power low when $k = 1$; yet our method shows high power in experiments. Their error control uses the connection between $k$-FWER and a negative binomial distribution, based on which we propose the estimator $\widehat{\text{FWER}}_t$ for our multi-step procedure, and prove the error control even when interaction is allowed. As far as we know, this estimator viewpoint of the FWER procedure is also new in the literature.

Jelle Goeman (private communication) pointed out that the \apt can be interpreted from the perspective of closed testing \citep{marcus1976closed}. Our method is also connected with the fallback procedure \citep{wiens2005fallback}, which allows for arbitrary dependence but is not interactive and combine covariate information with $p$-values to determine the ordering. See Appendix~\ref{apd:closed_testing} for details. 

\vspace{1cm}
\paragraph{The \apt in practice.}
Technically in a fully interactive procedure, a human can examine all the information in $\mathcal{F}_{t-1}$ and pick $i_t^*$ subjectively or by any other principle, but doing so for every step could be tedious and unnecessary. 
Instead, the analyst can design an automated version of the \apt, and still keeps the flexibility to change it at any iteration.  For example, the analyst can implement an automated algorithm to first exclude $80\%$ hypotheses (say). If $\widehat{\text{FWER}}_t$ is still larger than level $\alpha$, the analyst can pause the procedure manually to look at the unmasked $p$-value information, update her prior knowledge, and modify the current algorithm. The next section presents an automated implementation of the \apt that takes into account the structure on the non-nulls.

\section{An instantiation of an automated algorithm, and numerical experiments} \label{sec::sim}
One main advantage of the \apt is the flexibility to include prior knowledge and human guidance. The analyst might have an intuition about what structural constraints the non-nulls have. For example, we consider 
two structures: (a) a grid of hypotheses where the non-nulls are in a cluster (of some size/shape, at some location; see Figure~\ref{fig:true_cluster}), which is a reasonable prior belief when one wants to identify a tumor in a brain image; and (b) a tree of hypotheses where a child can be non-null only if its parent is non-null, as may be the case in applications involving wavelet decompositions.

\begin{figure}[t]
    \centering
    \hspace{1.5cm}
    \begin{subfigure}[t]{0.2\textwidth}
        \centering
        \includegraphics[width=1\linewidth]{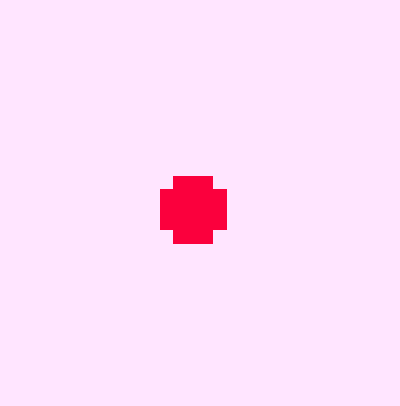}
        \caption{True non-nulls (21 hypotheses).}
        \label{fig:true_cluster}
    \end{subfigure}
    \hfill
    \begin{subfigure}[t]{0.2\textwidth}
        \centering
        \includegraphics[width=1\linewidth]{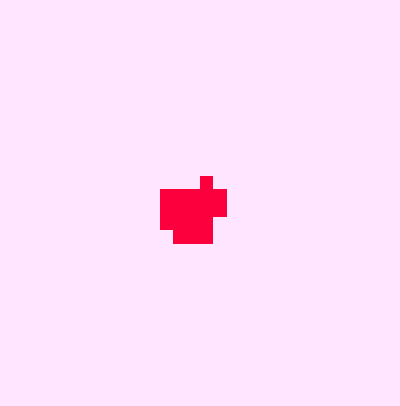}
        \caption{18 rejections by the \apt.}
    \end{subfigure}
    \hfill
    \begin{subfigure}[t]{0.2\textwidth}
        \centering
        \includegraphics[width=1\linewidth]{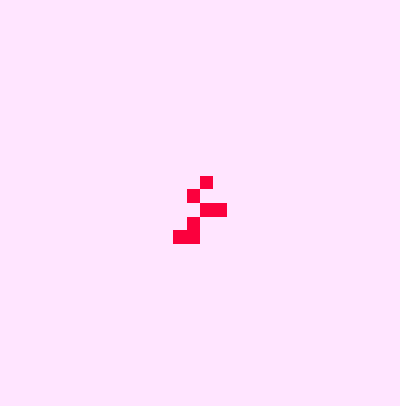}
        \caption{7 rejections by the {\v{S}}id{\'a}k correction}
    \end{subfigure}
    \hspace{1.5cm}
    \caption{An instance of rejections by the \apt and the {\v{S}}id{\'a}k correction \citep{vsidak1967rectangular}. Clustered non-nulls are simulated from the setting in Section~\ref{set:cluster} with a fixed alternative mean $\mu = 3$.}
    \label{fig:visual_cluster}
\end{figure}

\subsection{An example of an automated algorithm under clustered non-null structure} \label{sec::auto_alg}

We propose an automated algorithm of the \apt that incorporates the structure of clustered non-nulls. Here, the side information $x_i$ is the coordinates of each hypothesis~$i$. The idea is that at each step of excluding possible nulls, we peel off the boundary of the current $\mathcal{R}_t$, such that the rejection set stays connected (see Figure~\ref{fig:visual_exclude}). 

\begin{figure}[h!]
    \centering
    \hspace{1cm}
    \begin{subfigure}[t]{0.15\textwidth}
        \centering
        \includegraphics[ width=1\linewidth]{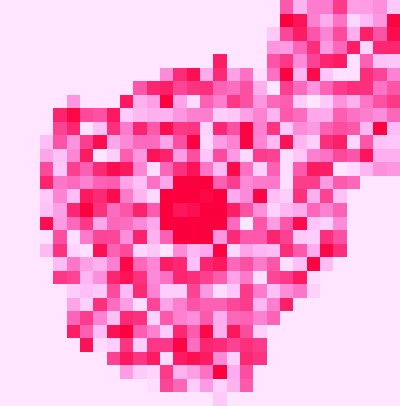}
    \end{subfigure}
    \hfill
    \begin{subfigure}[t]{0.15\textwidth}
        \centering
        \includegraphics[width=1\linewidth]{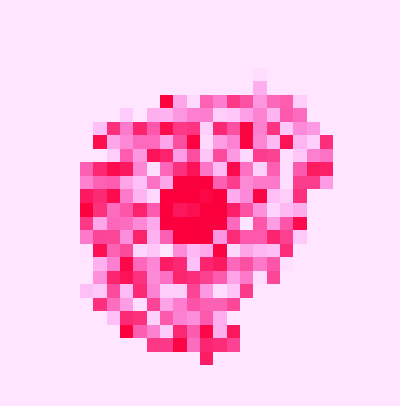}
    \end{subfigure}
    \hfill
    \begin{subfigure}[t]{0.15\textwidth}
        \centering
        \includegraphics[width=1\linewidth]{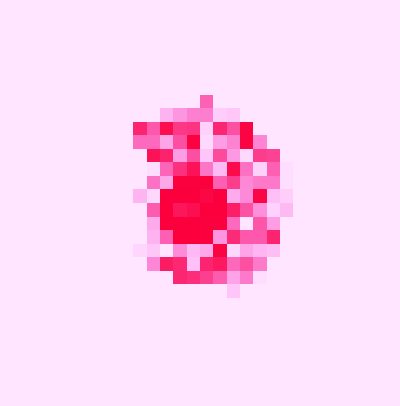}
    \end{subfigure}
    \hfill
    \begin{subfigure}[t]{0.15\textwidth}
        \centering
        \includegraphics[width=1\linewidth]{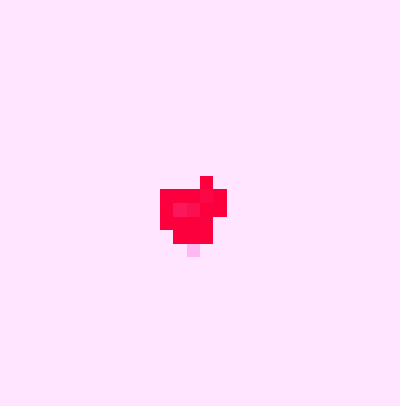}
    \end{subfigure}
    \hspace{1cm}
    \caption{An illustration of $\mathcal{R}_t$ generated by the automated algorithm described in Section~\ref{sec::auto_alg}, at $t = 50, 100, 150$ and $t = 220$ when the algorithm stops. The $p$-values in $\mathcal{R}_t$ are plotted.}
    \label{fig:visual_exclude}
\end{figure}

Suppose each hypothesis $H_i$ has a score $S_i$ to measure the likelihood of being non-null (\emph{non-null likelihood}). A simple example is $S_i = -g(P_i)$ since larger $g(P_i)$ indicates less chance of being a non-null (more details on $S_i$ to follow). We now describe an explicit fixed procedure to shrink $\mathcal{R}_t$. Given two parameters $d$ and $\delta$ (eg. $d = 5, \delta = 5\%$), it replaces step 1 and 2 in  Algorithm~\ref{alg:apt} as follows:
\begin{enumerate}
    \item[(a)] Divide $\mathcal{R}_{t-1}$ from its center to $d$ cones (like slicing a pizza); in each cone, consider a fraction $\delta$ of hypotheses farthest from the center, denoted $\smash{\mathcal{R}_{t-1}^1, \ldots, \mathcal{R}_{t-1}^d}$;
    \item[(b)]  Compute ${\bar S^j = \frac{1}{|\mathcal{R}_{t-1}^j|} \sum_{i \in \mathcal{R}_{t-1}^j} S_i}$ for $j = 1, \ldots, d$;
    \item[(c)] Update ${\mathcal{R}_t = \mathcal{R}_{t-1}\backslash \mathcal{R}_{t-1}^k}$, where ${k = \argmin_{j} \bar S^j}$.
\end{enumerate}

The score $S_i$ that estimates the non-null likelihood can be computed with the aid of a working statistical model. For example, consider a mixture model where each $p$-value $P_i$ is drawn from a mixture of a null distribution $F_0$ (eg: uniform) with probability $1 - \pi_i$ and an alternative distribution $F_1$ (eg: beta distribution) with probability $\pi_i$, or equivalently, 
\begin{align} \label{eq:s_model}
    P_i \overset{d}{=} (1 - \pi_i) F_0 + \pi_i F_1.
\end{align}
To account for the clustered structure of non-nulls, we may further assume a model that treats $\pi_i$ as a smooth function of the covariates $x_i$. The hidden missing bits $\{h(P_i)\}_{i\in R_t}$ can be inferred from $g(P_i)$ and the unmasked $h(P_i)$ by the EM algorithm (see details in Appendix~\ref{apd:em}). As $R_t$ shrinks, progressively unmasked missing bits improve the estimation of non-null likelihood and increase the power. Importantly, the FWER is controlled regardless of the correctness of the above model or any other heuristics to shrink $R_t$.

The above algorithm is only one automated example and there are many possibilities of what we can do to shrink $R_t$.
\begin{enumerate}
    \item A different algorithm can be developed for a different structure. For example, when hypotheses have a hierarchical structure and the non-nulls only appear on a subtree, an algorithm can gradually cut branches.
    
    \item The score $S_i$ for non-null likelihood is not exclusive for the above algorithm -- it can be used in any heuristics such as directly ordering hypotheses by $S_i$. 
    
    \item Human interaction can help the automated procedure: the analyst can stop and modify the automated algorithm at any iteration. It is a common case where prior knowledge might not be accurate, or there exist several plausible structures. The analyst may try different algorithms and improve their understanding of the data as the test proceeds. In the example of clustered non-nulls, the underlying truth might have two clustered non-nulls instead of one. After several iterations of the above algorithm that is designed for a single cluster, the shape of $\mathcal{R}_t$ could look like a dumbbell, so the analyst can split $\mathcal{R}_t$ into two subsets if they wish.
\end{enumerate}

Note that there is no universally most powerful test in nonparametric settings since we do not make assumptions on the distribution of non-null $p$-values, or how informative the covariates are. It is possible that the classical Bonferroni-Holm procedure \citep{holm1979simple} might have high power if applied with appropriate weights. Likewise, the power of our own test might be improved by changing the working model or choosing some other heuristic to shrink~$\mathcal{R}_t$. The main advantage of our method is that it can accommodate structural and covariate information and revise the modeling on the fly (as $p$-values are unmasked) while other methods commit to one type of structure without looking at the data.

Next, we demonstrate via experiments that the \apt can improve power over the {\v{S}}id{\'a}k correction, a baseline method that does not take side information into account\footnote{In all experiments, the Hommel method has similar power to the {\v{S}}id{\'a}k correction, and was hence omitted.}. We chose a clustered non-null structure for visualization and intuition, though our test can utilize any covariates, structural constraints, domain knowledge, etc.

\subsection{Numerical experiments for clustered non-nulls}
For most simulations in this paper, we use the setting below,
\begin{setting} \label{set:cluster}
Consider 900 hypotheses arranged in a $30 \times 30$ grid with a disc of 21 non-nulls. Each hypothesis tests the mean value of a univariate Gaussian as described in~\eqref{eq:normal_hp}. The true nulls are generated from $N(0,1)$ and non-nulls from $N(\mu,1)$, where we varied $\mu$ as $(1,2,3,4,5)$. For all experiments in the paper, the FWER control is set at level $\alpha = 0.2$, and the power is averaged over 500 repetitions\footnote{The standard error of FWER and averaged power are less than 0.02, thus ignored from the plots in this paper.}.
\end{setting}

\begin{figure*}[t]
    \centering
    \hspace{2cm}
    \begin{subfigure}[t]{0.32\textwidth}
        \centering
        \includegraphics[trim=0cm 0cm 0cm 0cm, clip, width=1\linewidth]{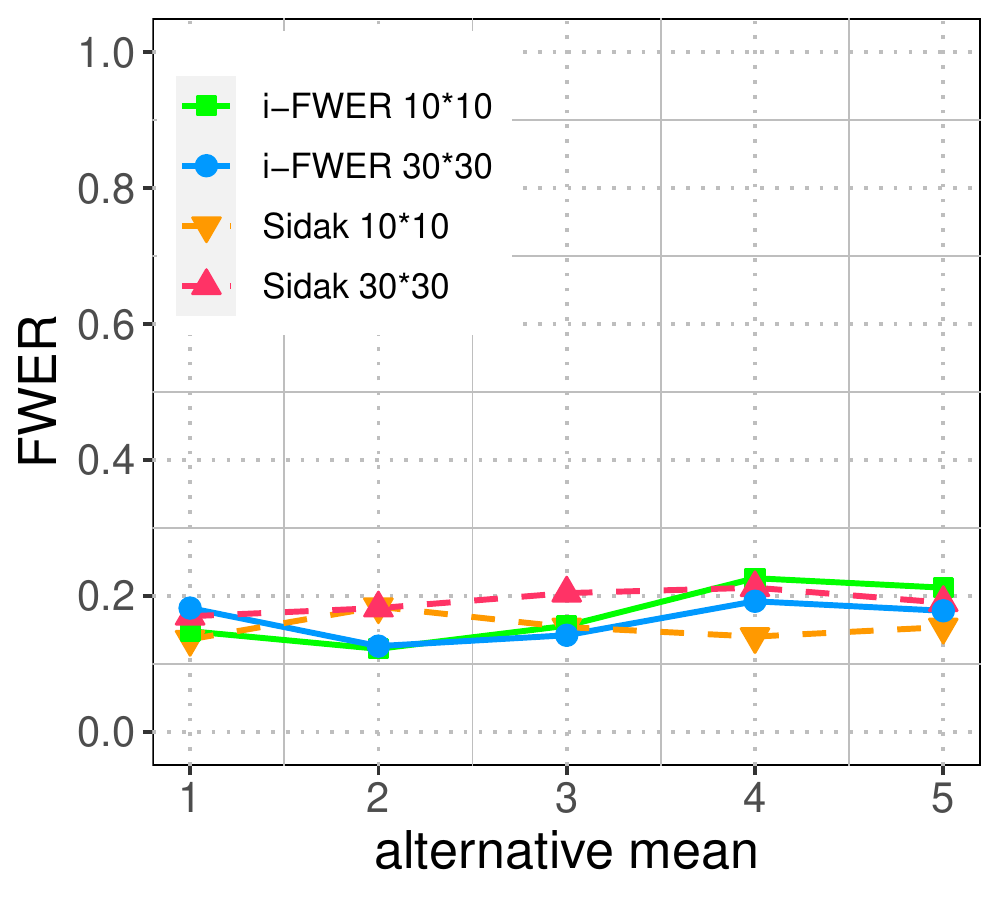}
    \end{subfigure}
    \hfill
    \begin{subfigure}[t]{0.32\textwidth}
        \centering
        \includegraphics[trim=0cm 0cm 0cm 0cm, clip, width=1\linewidth]{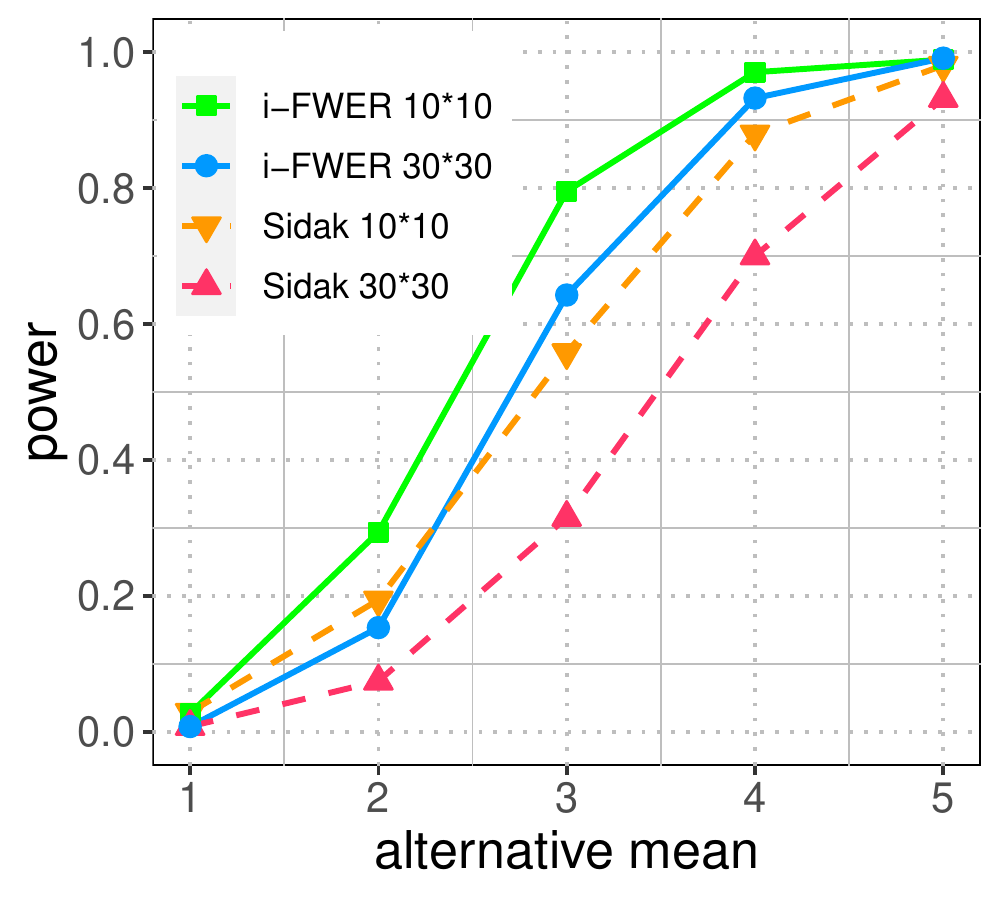}
    \end{subfigure}
    \hspace{2cm}
    \caption{The \apt versus {\v{S}}id{\'a}k for clustered non-nulls. The experiments are described in Section~\ref{set:cluster} where we tried two sizes of hypotheses grid: $10 \times 10$ and $30 \times 30$ (the latter is a harder problem since the number of nulls increases while the number of non-nulls remains fixed). Both methods show valid FWER control (left). The \apt has higher power under both grid sizes (right).}
    \label{fig:tent_power}
\end{figure*}

The \apt has higher power than the {\v{S}}id{\'a}k correction, which does not use the non-null structure (see Figure~\ref{fig:tent_power}). It is hard for most existing methods to incorporate the knowledge that non-nulls are clustered without knowing the position or the size of this cluster. By contrast, such information can be learned in the \apt by looking at the masked $p$-values and the progressively unmasked missing bits. This advantage of the \apt is more evident as the number of nulls increases (by increasing the grid size from ${10\times 10}$ to $30\times 30$ with the number of non-nulls fixed). Note that the power of both methods decreases, but the \apt seems less sensitive. This robustness to nulls is expected as the \apt excludes most nulls before rejection, whereas the {\v{S}}id{\'a}k correction treats all hypotheses equally.

\subsection{An example of an automated algorithm under a hierarchical structure of hypotheses}
When the hypotheses form a tree, the side information $x_i$ encodes the parent-child relationship (the set of indices of the children nodes for each hypothesis~$i$). Suppose we have prior knowledge that a node cannot be non-null if its parent is null, meaning that the non-nulls form a subtree with the same root. We now develop an automated algorithm that prunes possible nulls among the leaf nodes of current~$\mathcal{R}_t$, such that the rejection set has such a subtree shape. Like the algorithm for clustered non-nulls, we use a score $S_i$ to choose which leaf nodes to exclude. For example, the score $S_i$ can be the estimated non-null likelihood learned from model~\eqref{eq:s_model}, where we account for the hierarchical structure by further assuming a partial order constraint on $\pi_i$ that $\pi_i \geq \pi_j$ if $j \in x_i$ (i.e., $i$ is the parent of $j$).

We simulate a tree of five levels (the root has twenty children and three children
for each parent node after that) with 801 nodes in total and 7 of them being non-nulls. The non-nulls gather in one of the twenty subtrees of the root. Individual $p$-values are generated by the hypotheses of testing zero-mean Gaussian, same as for the clustered structure, where we varied the non-null mean values $\mu$ as $(1,2,3,4,5)$.

\begin{figure}[t]
    \centering
    \begin{subfigure}[t]{0.4\textwidth}
        \centering
        \includegraphics[width=0.7\linewidth]{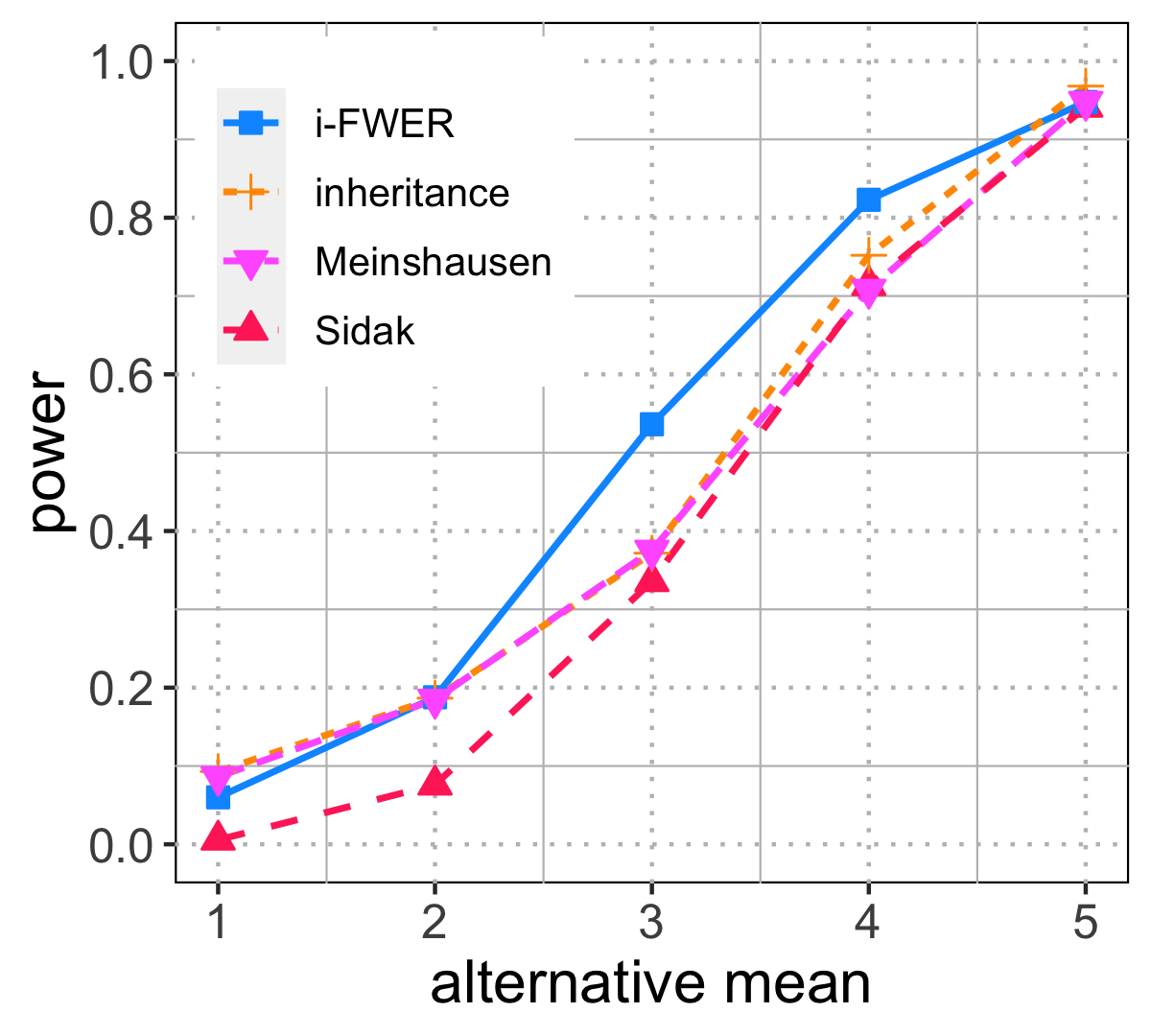}
    \end{subfigure}
    \caption{Power of the \apt under a tree structure when varying the alternative mean value. It has higher power than inheritance procedure, Meinshausen’s method, and the Sidak correction.}
    \label{fig:tree}
\end{figure}

In addition to the {\v{S}}id{\'a}k correction, we compare the \apt with two other methods for tree-structured hypotheses: Meinshausen’s method \citep{meinshausen2008hierarchical} and the inheritance procedure \citep{goeman2012inheritance}, which work under arbitrary dependence. Their idea is to pass the error budget from a parent node to its children in a prefixed manner, whereas our algorithm picks out the subtree with non-nulls based on the observed data. In our experiments, the \apt has the highest power (see Figure~\ref{fig:tree}).

\color{black}
The above results demonstrate the power of the \apt in one particular form where the masking is defined as~\eqref{eq:decomp_varyp}. However, any two functions that decompose the null $p$-values into two independent parts can, in fact, be used for masking and fit into the framework of the \apt (see the proofs of error control when using the following new masking functions in Appendix~\ref{apd:error_mask}). In the next section, we explore several choices of masking.

\section{New masking functions} \label{sec::mask}
Recall that masking is the key idea that permits interaction and controls error at the same time, by decomposing the $p$-values into two parts: masked $p$-value $g(P)$ and missing bits $h(P)$. Such splitting distributes the $p$-value information for two different purposes,  interaction and error control, leading to a tradeoff.
More information in $g(P)$ provides better guidance on how to shrink $\mathcal{R}_t$ and improves the power, while more information in $h(P)$ enhances the accuracy of estimating FWER and makes the test less conservative. This section explores several ways of masking and their influence on the power of the \apt. To distinguish different masking functions, we refer to masking~\eqref{eq:decomp_varyp} introduced at the very beginning as the ``tent'' function based on the shape of map~$g$ (see Figure~\ref{fig:tent}).

\begin{figure*}[h!]
    \centering
    \begin{subfigure}[t]{0.24\textwidth}
        \centering
        \includegraphics[trim=0cm 0.5cm 1cm 2cm, clip, width=1\linewidth]{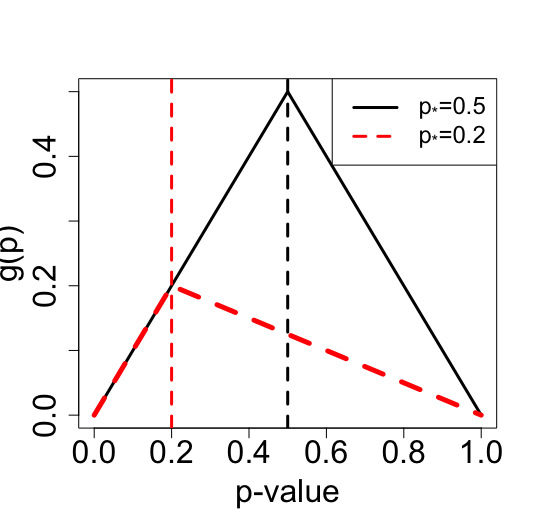}
        \caption{Tent functions when $p_*$ varies as $(0.5, 0.2)$. We need $p_* \leq \alpha$ for FWER control.}
        \label{fig:tent}
    \end{subfigure}
    \hfill
    \begin{subfigure}[t]{0.24\textwidth}
        \centering
        \includegraphics[trim=0cm 0.5cm 1cm 2cm, clip,width=1\linewidth]{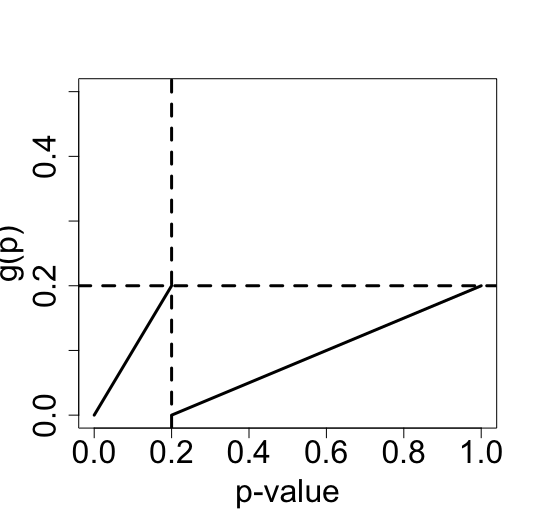}
        \caption{The railway function when $p_* = 0.2$.}
        \label{fig:railway}
    \end{subfigure}
    \hfill
    \begin{subfigure}[t]{0.24\textwidth}
        \centering
        \includegraphics[trim=0cm 0.5cm 1cm 2cm, clip,width=1\linewidth]{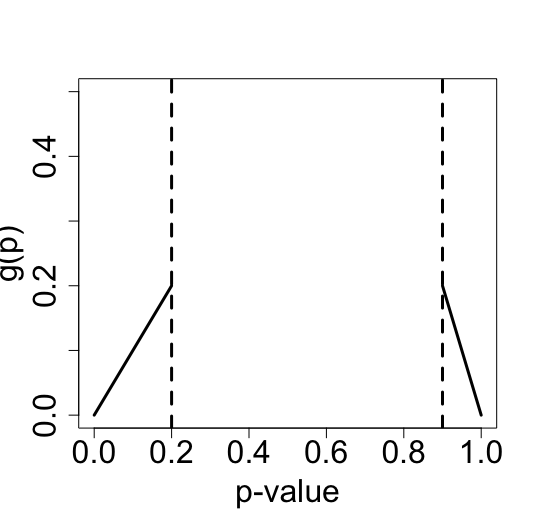}
        \caption{The gap function when ${(p_l, p_u) = (0.2, 0.9)}$.}
        \label{fig:gap}
    \end{subfigure}
    \hfill
    \begin{subfigure}[t]{0.24\textwidth}
        \centering
        \includegraphics[trim=0cm 0.5cm 1cm 2cm, clip,width=1\linewidth]{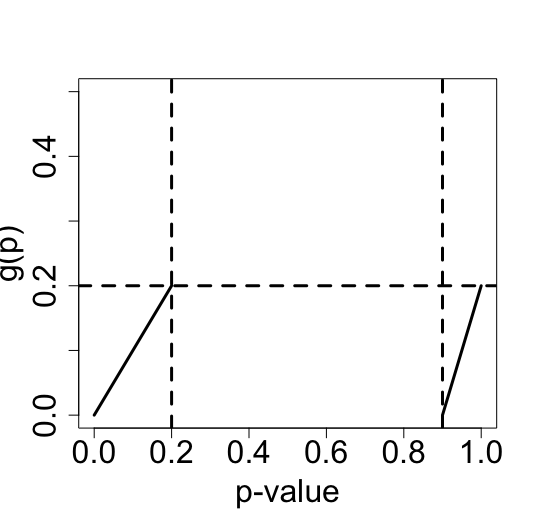}
        \caption{The gap-railway function when ${(p_l, p_u) = (0.2, 0.9)}$.}
        \label{fig:gap-railway}
    \end{subfigure}
    
    \caption{Different masking functions leaves different amount of information to $g(P)$ (and the complement part to $h(P)$).}
    \label{fig:masking}
\end{figure*}

\subsection{The ``railway'' function} \label{sec::railway}

We start with an adjustment to the tent function that flips the map $g$ when $p > p_*$, which we call the ``railway'' function (see Figure~\ref{fig:railway}). It does not change the information distribution between $g(P)$ and $h(P)$, and yet improves the power when nulls are conservative, as demonstrated later. 
Conservative nulls are often discussed under a general form of hypotheses testing for a parameter $\theta$:
\[
H_0: \theta \in \Theta_0 \quad \text{versus} \quad H_1: \theta \in \Theta_1,
\]
where $\Theta_0$ and $\Theta_1$ are two disjoint sets. Conservative nulls are those whose true parameter $\theta$ lies in the interior of $\Theta_0$. For example, when testing whether a Gaussian $N(\mu,1)$ has nonnegative mean in~\eqref{eq:normal_hp} where $\Theta_0 = \{\mu \leq 0\}$, the nulls are conservative when $\mu < 0$. The resulting $p$-values are biased toward larger values, which compared to the uniform $p$-values from nonconservative nulls should be easier to distinguish from that of non-nulls. However, most classical methods do not take advantage of it, but the \apt can, when using the railway function for masking:
\begin{align}
\label{eq:mask_railway}
h(P_i) =~& 2\cdot\one\{P_i < p_*\} - 1; \nonumber {}\\
\text{and } g(P_i) =~& \begin{cases}
P_i, & 0 \leq P_i < p_*,\\
\frac{p_*}{1 - p_*} (P_i - p_*),  & p_* \leq P_i \leq 1.
\end{cases}
\end{align}
The above masked $p$-value, compared with the tent masking~(\ref{eq:decomp_varyp}), can better distinguish the non-nulls from the conservative nulls. To see this, consider a $p$-value of 0.99. When $p_* = 0.2$, the masked $p$-value generated by the originally proposed tent function would be 0.0025, thus causing potential confusion with a non-null, whose masked $p$-value is also small. But the masked $p$-value from the railway function would be 0.1975, which is close to $0.2$, the upper bound of $g(P_i)$. Thus, it can easily be excluded by our algorithm.

We follow the setting in Section~\ref{set:cluster} for simulation 
, except that the alternative mean is fixed as $\mu = 3$, and the nulls are simulated from $N(\mu_0, 1)$, where the mean value $\mu_0$ is negative so that the resulting null $p$-values are conservative. We tried $\mu_0$ as $(0, -1, -2, -3, -4)$, with a smaller value indicating higher conservativeness, in the sense that the $p$-values are more likely to be biased to a larger value. When the null is not conservative ($\mu_0 = 0$), the \apt with the railway function and tent function have similar power. As the conservativeness of nulls increases, while the power of the \apt with the tent function decreases and the {\v{S}}id{\'a}k correction stays the same, the power of the \apt with the railway function increases (see Figure~\ref{fig:railway_power}). 

\begin{figure}[h!]
        \centering
        \includegraphics[trim=0cm 0cm 0cm 0cm, clip, width=0.35\linewidth]{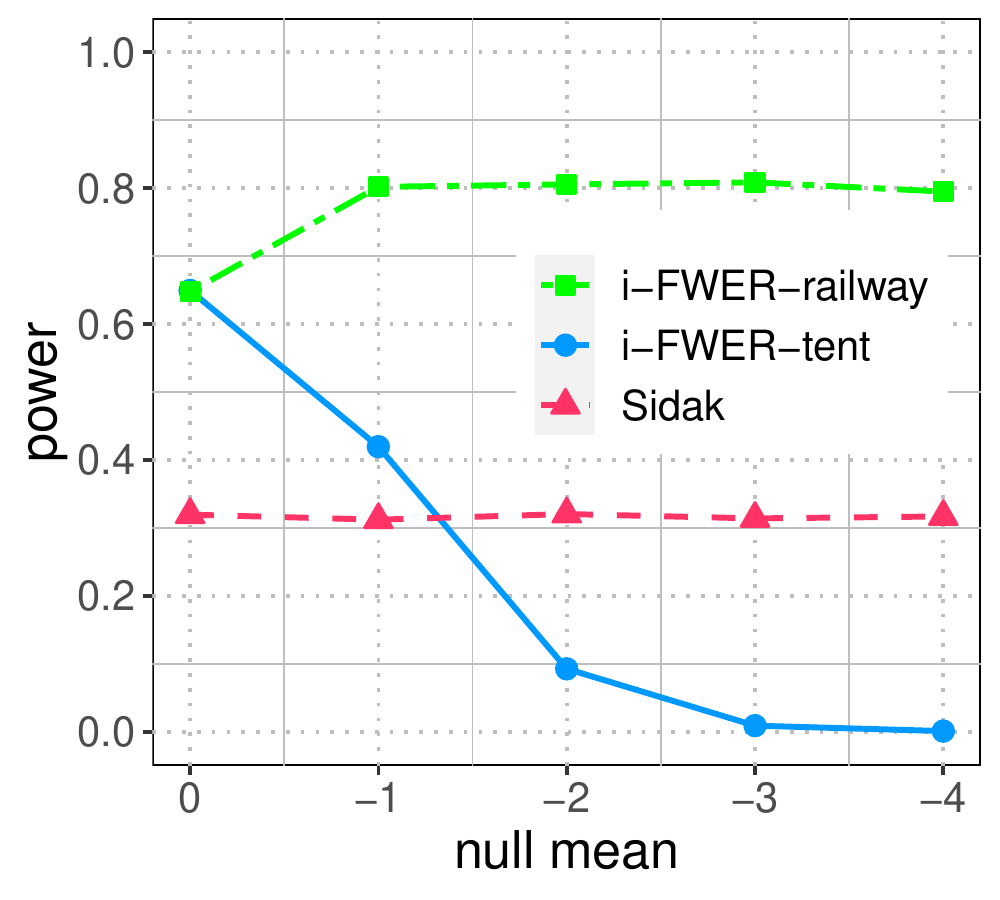}
    \caption{Power of the \apt with the tent function and the railway function, where the nulls become more conservative as the null mean decreases in $(0, -1, -2, -3, -4)$. The \apt benefits from conservative null when using the railway function.
    }
    \label{fig:railway_power}
\end{figure}

\subsection{The ``gap'' function}
Another form of masking we consider maps only the $p$-values that are close to $0$ or $1$, which is referred to as the ``gap'' function (see Figure~\ref{fig:gap}) 
. The resulting \apt directly unmasks all the $p$-values in the middle, and as a price, never rejects the corresponding hypotheses. Given two parameters~$p_l$ and~$p_u$, the gap function is defined as
\begin{align}
\label{eq:mask_gap}
h(P_i) =~& \begin{cases}
1, & 0 \leq P_i < p_l,\\
-1, & p_u < P_i \leq 1;
\end{cases} \nonumber {}\\
\text{and }g(P_i) =~& \begin{cases}
P_i, & 0 \leq P_i < p_l,\\
\frac{p_l}{1 - p_u} (1 - P_i),  & p_u < P_i \leq 1.
\end{cases}
\end{align}

All the $p$-values in $[p_l, p_u]$ are available to the analyst from the beginning. Specifically, let $\mathcal{M} = \{i: p_l < P_i < p_u\}$ be the set of skipped $p$-values in the masking step, then the available information at step $t$ for shrinking $\mathcal{R}_{t-1}$ is
\begin{align*}
    \mathcal{F}_{t-1} 
    := \sigma \Big(&\{x_i, g(P_i)\}_{i = 1}^n, \{P_i\}_{\{i \notin \mathcal{R}_{t-1}\}}, \{P_i\}_{\{i \in \mathcal{M}\}}\Big).
\end{align*}
The \apt with the gap masking changes slightly. We again consider two subsets of $\mathcal{R}_t$:
\begin{align*}
    \mathcal{R}_t^+ :=~& \{i\in \mathcal{R}_t: h(P_i) = 1\} \equiv  \{i\in \mathcal{R}_t: P_i < p_l\},{}\\
    \mathcal{R}_t^- :=~& \{i\in \mathcal{R}_t: h(P_i) = -1\} \equiv \{i\in \mathcal{R}_t: P_i > p_u\},
\end{align*}
and reject only the hypotheses in $\mathcal{R}_t^+$. The procedure of shrinking $\mathcal{R}_t$ stops when $\widehat{\text{FWER}}_t \leq \alpha$, where the estimation changes to
\begin{align} \label{eq:gap_threshold}
    \widehat{\text{FWER}}_t = 1 - \left(1 - \frac{p_l}{p_l + 1 - p_u}\right)^{|\mathcal{R}_t^-| + 1}.
\end{align}
To avoid the case that $\widehat{\text{FWER}}_t$ is always larger than $\alpha$ and the algorithm cannot make any rejection, the parameters $p_l$ and $p_u$ need to satisfy $\frac{1-\alpha}{\alpha}p_l + p_u < 1$. The above procedure boils down to the original \apt with the tent function when $p_l = p_u = p_*$.

The ``gap'' function reveals more information to select out possible nulls and help the analyst shrink $\mathcal{R}_t$, leading to power improvement in numerical experiments. We present the power results of the \apt using different masking functions after introducing a variant of the gap function.

\begin{figure}[t]
        \centering
        \includegraphics[trim=0cm 0cm 0cm 0cm, clip, width=0.35\linewidth]{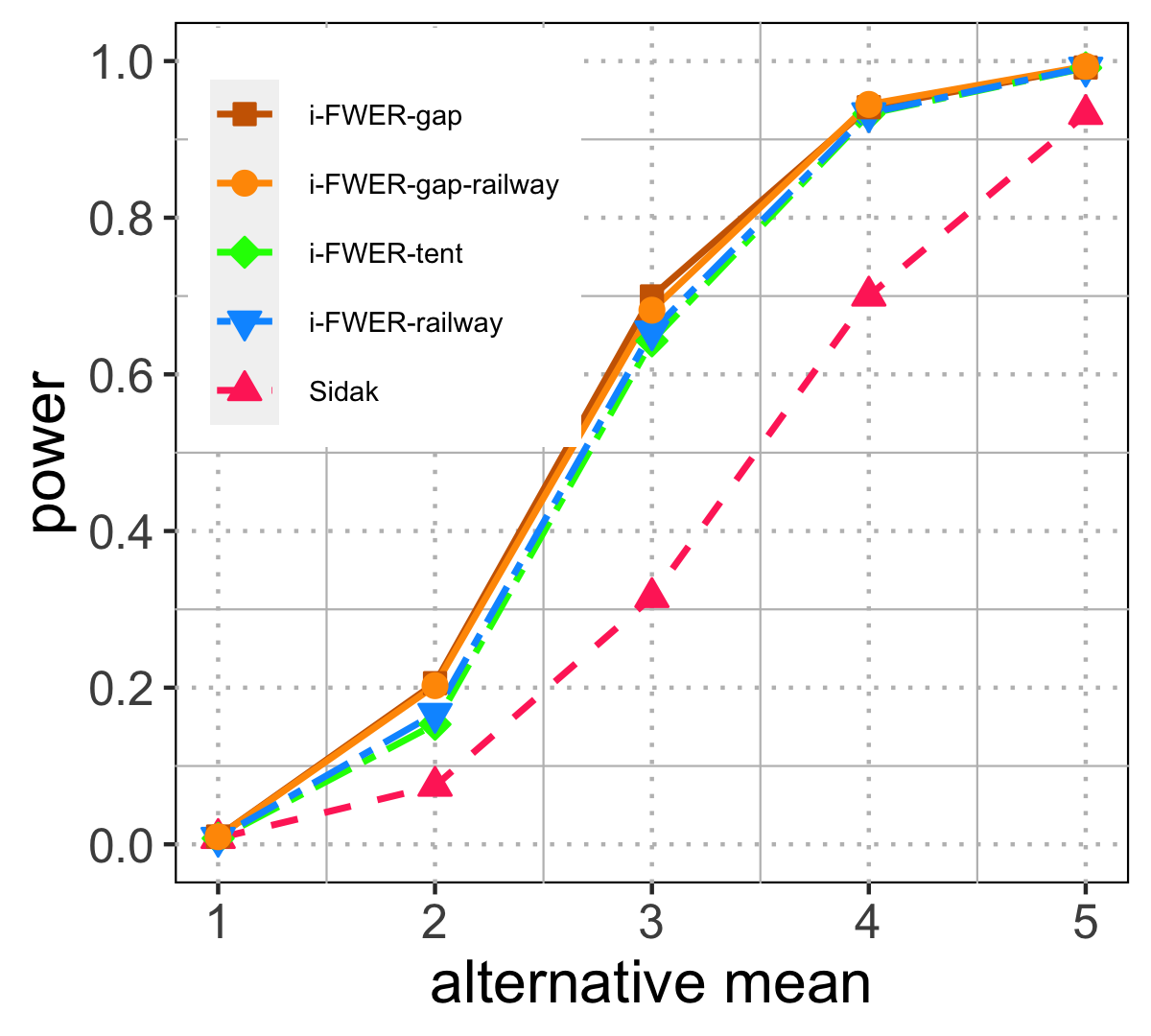}
    \caption{Power of the \apt with the tent function ($p_* = 0.1$) and the gap function ($p_l = 0.1, p_u = 0.5$). The gap function leads to slight improvement in power. Simulation follows the setting in Section~\ref{set:cluster}.}
    \label{fig:power_gap}
\end{figure}

\subsection{The ``gap-railway'' function}
Combining the idea of the gap and railway functions, we develop the ``gap-railway'' function such that the middle $p$-values are directly unmasked and the map $g$ for large $p$-values is an increasing function (see Figure~\ref{fig:gap-railway}). 
Given parameters~$p_l$ and~$p_u$, the gap-railway function is defined as
\begin{align}
\label{eq:mask_gap_railway}
h(P_i) =~& \begin{cases}
1, & 0 \leq P_i < p_l,\\
-1, & p_u < P_i \leq 1;
\end{cases} \nonumber {}\\
\text{and }g(P_i) =~& \begin{cases}
P_i, & 0 \leq P_i < p_l,\\
\frac{p_l}{1 - p_u} (P_i - p_u),  & p_u < P_i \leq 1.
\end{cases}
\end{align}

Comparing with the tent function with $p_* = p_l$, the \apt using the gap function additionally uses the entire $p$-values in $[p_l, p_u]$ for interaction, which leads to an increased power (see Figure~\ref{fig:power_gap}). The same pattern is maintained when we flip the mappings for large $p$-values, shown in the comparison of the railway function and the gap-railway function\footnote{The tests with the tent function and the railway function have similar power; and same for the gap function and the gap-railway function. As the null $p$-values follow an exact uniform distribution, so flipping the map $g$ for large $p$-values does not change the power.}.
This improvement also motivates why the \apt progressively unmasks $h(P_i)$, in other words, to reveal as much information to the analyst as allowed at the current step. Unmasking the $p$-values even for the hypotheses outside of the rejection set can improve the power, because they help the joint modeling of all the $p$-values, especially when there is some non-null structure.

To summarize, we have presented four types of masking functions: tent, railway, gap, gap-railway (see Figure~\ref{fig:masking}). Compared to the tent (gap) function, the railway (gap-railway) functions are more robust to conservative nulls. Compared with the tent (railway) function, the gap (gap-railway) function reveals more information to guide the shrinkage of $\mathcal{R}_t$. Note however that the railway or gap function is not \emph{always} better than the tent function. We may favor the tent function over the railway function when there are less $p$-values close to one, and we may favor the tent function over the gap function when there is considerable prior knowledge to guide the shrinkage of $\mathcal{R}_t$. 

The above discussion has explored specific non-null structures and masking functions. A large variety of masking functions and their advantages are yet to be discovered.

\section{A prototypical application to genetic data} \label{sec:real_data}
Below, we further demonstrate the power of the \apt using a real `airway dataset', which is analyzed by Independent Hypothesis Weighting (IHW) \citep{ignatiadis2016data} and AdaPT \citep{lei2018adapt}; these are (respectively) adaptive and interactive algorithms with FDR control for independent hypotheses. We compare the number of rejections made by a variant of the IHW with FWER control and the \apt using the tent function with the masking parameter $p_*$ chosen as $\alpha/20, \alpha/10, \alpha/2$, when the targeted FWER level $\alpha$ varies in $(0.1, 0.2, 0.3)$.

The airway data is an RNA-Seq dataset targeting the identification of differentially expressed genes in airway smooth muscle cell lines in response to dexamethasone, which contains 33469 genes (hypotheses) and a univariate covariate (the logarithm of normalized sample size) for each gene. The \apt makes more rejections than the IHW for all considered FWER levels and all considered choices of $p_*$ (see Table~\ref{tab:airway}). 

\begin{table}[ht]
\caption{Number of rejections by IHW and \apt under different FWER levels.}
\label{tab:airway}
\centering
\begin{tabular}{|c|c|c|c|c|}
\hline
\multirow{2}{*}{level $\alpha$} & \multirow{2}{*}{IHW} & \multicolumn{3}{c|}{i-FWER}\\ \cline{3-5}
 &  & $p_* = \alpha/2$ & $p_* = \alpha/10$ & $p_* = \alpha/20$ \\\hline
0.1 & 1552 & 1613 & 1681 & 1646\\ \hline
0.2 & 1645 & 1740 & 1849 & 1789\\ \hline
0.3 & 1708 & 1844 & 1925 & 1894\\ \hline
\end{tabular}
\end{table}

In hindsight, a small value for the masking parameter was more powerful in this dataset because over 1600 $p$-values are extremely small ($< 10^{-5}$), and these are highly likely to be the non-nulls. Thus, even when the masked $p$-values for all hypotheses are in a small range, such as $(0,0.01)$ when $\alpha = 0.1$ and $p_* = \alpha/10$, the $p$-values from the non-nulls still stand out because they gather below $10^{-5}$. At the same time, the smaller the $p_*$, the more accurate (less conservative) is our estimate of FWER in \eqref{eq:fwer_hat}; the algorithm can stop shrinking $\mathcal{R}_t$ earlier since more hypotheses with negative $h(P)$ are allowed to be included in the final $\mathcal{R}_t$. In practice, the choice of masking parameter can be guided by the prior belief of the strength of non-null signals: if the non-nulls have strong signal and hence extremely small $p$-values (such as the mean value $\mu \geq 5$ when testing if a univariate Gaussian has zero mean), a small masking parameter is preferred; otherwise, we recommend $\alpha/2$ to leave more information for interactively shrinking the rejection set $\mathcal{R}_t$.

\color{black}
\section{Discussion}

We proposed a multiple testing method with a valid FWER control while granting the analyst freedom of interacting with the revealed data. The masking function must be fixed in advance, but during the procedure of excluding possible nulls, the analyst can employ any model, heuristic, intuition, or domain knowledge, tailoring the algorithm to various applications. Although the validity requires an independence assumption, our method is a step forward to fulfilling the practical needs of allowing interactive human guidance to automated large-scale testing using ML in the sciences.

The critical idea that guarantees the FWER control is ``masking and unmasking''. 
A series of interactive tests are developed following the idea of masking: \citet{lei2018adapt} and \citet{lei2017star} proposed the masking idea and an interactive test with FDR control; \citet{duan2019interactive} developed an interactive test for the global null; this work presents an interactive test with FWER control. 
At a high level, masking-based interactive testing achieves rigorous conclusions in an exploratory framework, giving this broad technique much appeal and potential.

\section*{Code and Data}
Code can be found in \href{https://github.com/duanby/i-FWER}{https://github.com/duanby/i-FWER}. It was tested on macOS using R (version 3.6.0) and the following packages: magrittr, splines, robustbase, ggplot2.

Data in Section~\ref{sec:real_data} is collected by \cite{himes2014rna} and available in R package \texttt{airway}. We follow \cite{ignatiadis2016data} and \cite{lei2018adapt} to analyze the data using \texttt{DEseq2} package \citep{love2014moderated}.

\section*{Acknowledgements}
We thank Jelle Goeman for his insightful comments on the connection between our proposed method and closed testing. We thank Will Fithian for related discussions. Eugene Katsevich, Ian Waudby-Smith, Jinjin Tian and Pratik Patil are acknowledged for their feedback on an early draft, and anonymous reviewers for helpful suggestions.

\bibliography{ref}
\bibliographystyle{chicago}

\newpage
\appendix
\section{Distribution of the null $p$-values} \label{apd:mirror-cons}
With tent masking, error control holds for null $p$-values whose distribution satisfies a property called \emph{mirror-conservativeness}:
\begin{equation}
\label{cond:mirror_consv}
f(a) \leq f\left(1 - \frac{1 - p_*}{p_*} a\right), \quad \text{ for all } 0 \leq a \leq p_*,
\end{equation}
where $f$ is the probability mass function of $P$ for discrete $p$-values or the density function otherwise, and $p_*$ is the parameter in Algorithm~\ref{alg:apt} (see proof in Appendix~\ref{apd:thm1}). The mirror-conservativeness is first proposed by \citet{lei2018adapt} in the case of $p_* = 0.5$. A more commonly used notion of conservativeness is that $p$-values are stochastically larger than uniform:
\[
\mathbb{P}(P \leq a) \leq a, \quad \text{ for all } 0 \leq a \leq 1,
\]
which neither implies nor is implied by the mirror-conservativeness. 

A sufficent condition of the mirror-conservativeness is that the $p$-values have non-decreasing densities. For example, consider a one-dimensional exponential family and the hypotheses to test the value of its parameter~$\theta$:
\[
H_0: \theta \leq \theta_0, \quad \text{ versus } \quad H_1: \theta > \theta_0,
\]
where $\theta_0$ is a prespecified constant. The $p$-value calculated from the uniformly most powerful test is shown to have a nondecreasing density \citep{zhao2019multiple}; thus, it satisfies the mirror-conservativeness.
The conservative nulls described in Section~\ref{sec::railway} also fall into the above category where the exponential family is Gaussian, and the parameter is the mean value. Indeed, when the $p$-values have non-decreasing densities, the \apt also has a valid error control using alternative masking functions as proposed in Section~\ref{sec::mask} (see proof in Appendix~\ref{apd:error_mask}).

\section{Proof of Theorem~\ref{thm:fwer}} \label{apd:thm1}
The main idea of the proof is that the missing bits $h(P_i)$ of nulls are coin flips with probability $p_*$ to be heads, so the number of false rejections (i.e. the number of nulls with $h(P_i) = 1$ before the number of hypotheses with $h(P_i) = -1$ reaches a fixed number) is stochastically dominated by a negative binomial distribution. There are two main challenges. First, the interaction uses unmasked $p$-value information to reorder $h(P_i)$, so it is not trivial to show that the reordered $h(P_i)$ preserve the same distribution as that before ordering. Second, our procedure runs backward to find the first time that the number of hypotheses with negative $h(P_i)$ is below a fixed number, which differs from the standard description of a negative binomial distribution.

\subsection{Missing bits after interactive ordering}
We first study the effect of interaction. Imagine that Algorithm~\ref{alg:apt} does not have a stopping rule and generates a full sequence of $\mathcal{R}_t$ for $t = 0, 1, \ldots n$, where $\mathcal{R}_0 = [n]$ and $\mathcal{R}_n = \emptyset$. It leads to an ordered sequence of $h(P_i)$: 
\[
h(P_{\pi_1}), h(P_{\pi_2}), \ldots, h(P_{\pi_n}),
\]
where $\pi_n$ is the index of the first excluded hypothesis and $\pi_j$ denotes the index of the hypothesis excluded at step $n - j + 1$, that is $\pi_{j} = \mathcal{R}_{n - j} \backslash \mathcal{R}_{n - j + 1}$.
\begin{lemma} \label{lm:seq}
Suppose the null $p$-values are uniformly distributed and all the hypotheses are nulls, then for any $j = 1, \ldots, n$,
\[
\mathbb{E}\left[\one\left(h(P_{\pi_j}) = 1\right)\right] = p_*,
\]
and $\{\one\left(h(P_{\pi_j}) = 1\right)\}_{j=1}^n$ are mutually independent.
\end{lemma}
\begin{proof}
Recall that the available information for the analyst to choose $\pi_j$ is ${\mathcal{F}_{n-j} = \sigma \Big(\{x_i, g(P_i)\}_{i=1}^n, \{P_i\}_{i \notin \mathcal{R}_{n - j}}\Big)}$. First, consider the conditional expectation:
\begin{align} \label{eq:conditional_case1}
    &\mathbb{E}\left[\one\left(h(P_{\pi_j}) = 1\right) \middle| \mathcal{F}_{n-j}\right]\nonumber{}\\
    =~& \sum_{i \in [n]} \mathbb{E}\left[\one\left(h(P_{\pi_j}) = 1\right) \middle| \pi_j = i, \mathcal{F}_{n-j}\right] \mathbb{P}\left(\pi_j = i \middle| \mathcal{F}_{n-j}\right) \nonumber{}\\
    \overset{(a)}{=}~& \sum_{i \in \mathcal{R}_{n - j}} \mathbb{E}\left[\one\left(h(P_{i}) = 1\right) \middle| \pi_j = i, \mathcal{F}_{n-j}\right] \mathbb{P}\left(\pi_j = i \middle| \mathcal{F}_{n-j}\right) \nonumber{}\\
    \overset{(b)}{=}~& \sum_{i \in \mathcal{R}_{n - j}} \mathbb{E}\left[\one\left(h(P_{i}) = 1\right) \middle| \mathcal{F}_{n-j}\right] \mathbb{P}\left(\pi_j = i \middle| \mathcal{F}_{n-j}\right)  \nonumber{}\\
    \overset{(c)}{=}~& \sum_{i \in \mathcal{R}_{n - j}} \mathbb{E}\left[\one\left(h(P_{i}) = 1\right)\right] \mathbb{P}\left(\pi_j = i \middle| \mathcal{F}_{n-j}\right)  \nonumber{}\\
     =~& p_* \sum_{i \in \mathcal{R}_{n - j}} \mathbb{P}\left(\pi_j = i \middle| \mathcal{F}_{n-j}\right) = p_*,
\end{align}
where equation~$(a)$ narrows down the choice of $i$ because $\mathbb{P}(\pi_j = i \mid \mathcal{F}_{n-j}) = 0$ for any $i \notin  \mathcal{R}_{n - j}$; equation~$(b)$ drops the condition of $\pi_j = i$ because $\pi_j$ is measurable with respect to $\mathcal{F}_{n-j}$; and equation~$(c)$ drops the condition $\mathcal{F}_{n-j}$ because by the independence assumptions in Theorem~\ref{thm:fwer}, $h(P_i)$ is independent of $\mathcal{F}_{n-j}$ for any $i \in \mathcal{R}_{n - j}$. 

Therefore, by the law of iterated expectations, we prove the claim on expected value:
\[
\mathbb{E}\left[\one\left(h(P_{\pi_j}) = 1\right)\right] = \mathbb{E}\left[ \mathbb{E}\left[\one\left(h(P_{\pi_j}) = 1\right) \middle| \mathcal{F}_{n-j}\right]\right] = p_*.
\]

For mutual independence, we can show that for any $1 \leq k < j \leq n$, $\one\left(h(P_{\pi_k}) = 1\right)$ is independent of $\one\left(h(P_{\pi_j}) = 1\right)$. Consider the conditional expectation:
\begin{align*}
    &\mathbb{E}\left[\one\left(h(P_{\pi_k}) = 1\right) \middle| \one\left(h(P_{\pi_j}) = 1\right)\right] {}\\
    =~& \mathbb{E}\left[\mathbb{E}\left[\one\left(h(P_{\pi_k}) = 1\right) \middle| \mathcal{F}_{n-k}, \one\left(h(P_{\pi_j}) = 1\right)\right]\middle| \one\left(h(P_{\pi_j}) = 1\right)\right]{}\\
    &\text{(note that } \one\left(h(P_{\pi_j}) = 1\right) \text{ is measurable with respect to } \mathcal{F}_{n-k} \text{)}{}\\
    =~& \mathbb{E}\left[\mathbb{E}\left[\one\left(h(P_{\pi_k}) = 1\right) \middle| \mathcal{F}_{n-k}\right]\middle| \one\left(h(P_{\pi_j}) = 1\right)\right]{}\\
    &\text{(use equation~\eqref{eq:conditional_case1} for the conditional expectation)}{}\\
    =~& \mathbb{E}\left[p_* \mid \one\left(h(P_{\pi_j}) = 1\right)\right] = p_*.
\end{align*}
It follows that $\one\left(h(P_{\pi_k}) = 1\right) \mid \one\left(h(P_{\pi_j}) = 1\right)$ is a Bernoulli with parameter $p_*$, same as the marginal distribution of $\one\left(h(P_{\pi_k}) = 1\right)$; thus, $\one\left(h(P_{\pi_k}) = 1\right)$ is independent of $\one\left(h(P_{\pi_j}) = 1\right)$ for any $1 \leq k < j \leq n$ as stated in the Lemma.
\end{proof}

\begin{corollary} \label{col:seq}
Suppose the null $p$-values are uniformly distributed and there may exist non-nulls. For any $j = 1, \ldots, n$,
\[
\mathbb{E}\left[\one\left(h(P_{\pi_j}) = 1 \right)\middle| \left\{\one\left(h(P_{\pi_k}) = 1 \right)\right\}_{k = j + 1}^n, \left\{\one\left(\pi_k \in \mathcal{H}_0\right) \right\}_{k = j + 1}^n, \pi_j \in \mathcal{H}_0\right] = p_*,
\]
where $\{\pi_k\}_{k=j+1}^n$ represents the hypotheses excluded before $\pi_j$. 
\end{corollary}
\begin{proof}
Denote the condition $\sigma\left(\left\{\one\left(h(P_{\pi_k}) = 1 \right)\right\}_{k = j + 1}^n, \left\{\one\left(\pi_k \in \mathcal{H}_0\right) \right\}_{k = j + 1}^n\right)$ as~$\mathcal{F}_{n - j}^h$. The proof is similar to Lemma~\ref{lm:seq}. First, consider the expectation conditional on $\mathcal{F}_{n-j}$:
\begin{align} \label{apd:eq_tbc}
    &\mathbb{E}\left[\one\left(h(P_{\pi_j}) = 1 \right) \middle| \mathcal{F}_{n - j}^h, \pi_j \in \mathcal{H}_0, \mathcal{F}_{n-j}\right]\nonumber{}\\
    =~&\mathbb{E}\left[\one\left(h(P_{\pi_j}) = 1 \right) \middle| \pi_j \in \mathcal{H}_0, \mathcal{F}_{n-j}\right] \quad (\text{since } \mathcal{F}_{n - j}^h \text{ is a subset of } \mathcal{F}_{n-j})\nonumber{}\\
    =~& \sum_{i \in [n]} \mathbb{E}\left[\one\left(h(P_{i}) = 1 \right) \mid \pi_j = i, \pi_j \in \mathcal{H}_0, \mathcal{F}_{n-j}\right] \mathbb{P}(\pi_j = i \mid \pi_j \in \mathcal{H}_0, \mathcal{F}_{n-j}) \nonumber{}\\
    =~& \sum_{i \in \mathcal{R}_{n - j} \cap \mathcal{H}_0} \mathbb{E}\left[\one\left(h(P_{i}) = 1 \right) \mid \pi_j = i, \pi_j \in \mathcal{H}_0, \mathcal{F}_{n-j}\right] \mathbb{P}(\pi_j = i \mid \pi_j \in \mathcal{H}_0, \mathcal{F}_{n-j}) \nonumber{}\\
    =~& \sum_{i \in \mathcal{R}_{n - j} \cap \mathcal{H}_0} \mathbb{E}\left[\one\left(h(P_{i}) = 1 \right) \middle| \mathcal{F}_{n-j}\right] \mathbb{P}(\pi_j = i \mid \pi_j \in \mathcal{H}_0, \mathcal{F}_{n-j}) \nonumber{}\\
    =~& p_* \sum_{i \in \mathcal{R}_{n - j} \cap \mathcal{H}_0} \mathbb{P}(\pi_j = i \mid \pi_j \in \mathcal{H}_0, \mathcal{F}_{n-j}) = p_*,
\end{align}
where we use the same technics of proving equation~\eqref{eq:conditional_case1}.

Thus, by the law of iterated expectations, we have
\begin{align*}
    &\mathbb{E}\left[\one\left(h(P_{\pi_j}) = 1 \right)\middle| \mathcal{F}_{n - j}^h, \pi_j \in \mathcal{H}_0\right]{}\\
    =~& \mathbb{E}\left[ \mathbb{E}\left[\one\left(h(P_{\pi_j}) = 1 \right) \middle| \mathcal{F}_{n - j}^h, \pi_j \in \mathcal{H}_0, \mathcal{F}_{n-j}\right] \middle| \mathcal{F}_{n - j}^h, \pi_j \in \mathcal{H}_0\right] = p_*,
\end{align*}
which completes the proof.
\end{proof}

\begin{corollary} \label{col:seq_cons}
Suppose the null $p$-values can be mirror-conservative as defined in~\eqref{cond:mirror_consv} and there may exist non-nulls, then for any~$j = 1, \ldots, n$,
\[
\mathbb{E}\left[\one\left(h(P_{\pi_j}) = 1\right) \middle| \left\{\one\left(h(P_{\pi_k}) = 1 \right)\right\}_{k = j + 1}^n, \left\{\one\left(\pi_k \in \mathcal{H}_0\right) \right\}_{k = j + 1}^n, \pi_j \in \mathcal{H}_0,  \{g(P_{\pi_k})\}_{k=1}^n\right] \leq p_*,
\]
where $\{g(P_{\pi_k})\}_{k=1}^n$ denotes $g(P)$ for all the hypotheses (excluded or not).
\end{corollary}
\begin{proof}
First, we claim that a mirror-conservative $p$-value $P$ satisfies that
\begin{align} \label{eq:cond_cons}
    \mathbb{E}\left[\one\left(h(P) = 1\right) \mid g(P)\right] \leq p_*,
\end{align}
since for every $a \in (0,p_*)$,
\begin{align*}
    &\mathbb{E}\left[\one\left(h(P) = 1\right) \mid g(P) = a\right]{}\\
    =~& \frac{p_* f(a)}{ p_* f(a)  + (1 - p_*) f\left(1 - \frac{1 - p_*}{p_*} a\right) }{}\\
    =~& \frac{p_* }{ p_* + (1 - p_*) f\left( 1 - \frac{1 - p_*}{p_*} a\right)/ f(a) } \leq p_*,
\end{align*}
where recall that $f$ is the probability mass function of $P$ for discrete $p$-values or the density function otherwise. The last inequality comes from the definition of mirror-conservativeness in~\eqref{cond:mirror_consv}.
The rest of the proof is similar to Corollary~\ref{col:seq}, where we first condition on $\mathcal{F}_{n-j}$:
\begin{align*}
    &\mathbb{E}\left[\one\left(h(P_{\pi_j}) = 1\right) \middle| \mathcal{F}_{n-j},\mathcal{F}_{n - j}^h, \pi_j \in \mathcal{H}_0, \{g(P_{\pi_k})\}_{k=1}^n\right]  \nonumber{}\\
    =~& \sum_{i \in \mathcal{R}_{n - i} \cap \mathcal{H}_0} \mathbb{E}\left[\one\left(h(P_i) = 1\right) \mid \mathcal{F}_{n-j}\right] \mathbb{P}\left(\pi_j = i \middle| \mathcal{F}_{n-j}, \mathcal{F}_{n - j}^h, \pi_j \in \mathcal{H}_0, \{g(P_{\pi_k})\}_{k=1}^n \right) \nonumber{}\\
    \overset{(a)}{=}~& \sum_{i \in \mathcal{R}_{n - i} \cap \mathcal{H}_0} \mathbb{E}\left[\one\left(h(P_i) = 1\right) \mid g(P_i)\right] \mathbb{P}\left(\pi_j = i \middle| \mathcal{F}_{n-j}, \mathcal{F}_{n - j}^h, \pi_j \in \mathcal{H}_0, \{g(P_{\pi_k})\}_{k=1}^n \right) \nonumber{}\\
    \leq~& p_* \sum_{i \in \mathcal{R}_{n - i} \cap \mathcal{H}_0} \mathbb{P}\left(\pi_j = i \middle| \mathcal{F}_{n-j}, \mathcal{F}_{n - j}^h, \pi_j \in \mathcal{H}_0, \{g(P_{\pi_k})\}_{k=1}^n \right) = p_*,
\end{align*}
where equation~$(a)$ simplify the condition of $\mathcal{F}_{n-j}$ to $g(P_i)$ because for any $i \in \mathcal{R}_{n - i} \cap \mathcal{H}_0$, $h(P_i)$ is independent of other information in $\mathcal{F}_{n-j}$.

Then, by the law of iterated expectations, we obtain
\begin{align*}
    &\mathbb{E}\left[\one\left(h(P_{\pi_j}) = 1\right) \middle| \mathcal{F}_{n - j}^h, \pi_j \in \mathcal{H}_0, \{g(P_{\pi_k})\}_{k=1}^n \right]{}\\
    =~& \mathbb{E}\left[ \mathbb{E}\left[\one\left(h(P_{\pi_j}) = 1\right) \middle| \mathcal{F}_{n-j}, \mathcal{F}_{n - j}^h, \pi_j \in \mathcal{H}_0,  \{g(P_{\pi_k})\}_{k=1}^n\right] \middle| \mathcal{F}_{n - j}^h, \pi_j \in \mathcal{H}_0, \{g(P_{\pi_k})\}_{k=1}^n\right] \leq p_*,
\end{align*}
thus the proof is completed.
\end{proof}

\subsection{Negative binomial distribution}
In this section, we discuss several procedures for Bernoulli trials (coin flips) and their connections with the negative binomial distribution.

\begin{lemma} \label{lm:basic_nb}
Suppose $A_1, \ldots, A_n$ are i.i.d. Bernoulli with parameter $p_*$. For $t = 1, \ldots, n$, consider the sum ${M_t = \sum_{j=1}^t A_j}$ and the filtration $\mathcal{G}_t^o = \sigma\left(\{A_j\}_{j=1}^{t}\right)$. Define a stopping time parameterized by a constant $v (\geq 1)$:
\begin{align} \label{eq:stop_nb}
    \tau^o  = \min\{0 < t \leq n: t - M_t \geq v \text{ or } t = n\},
\end{align}
then $M_{\tau^o}$ is stochastically dominated by a negative binomial distribution:
\[
M_{\tau^o} \preceq \mathrm{NB}(v, p_*).
\]
\end{lemma}
\begin{proof}
Recall that the negative binomial $\mathrm{NB}(v,p_*)$ is the distribution of the number of success in a sequence of independent and identically distributed Bernoulli trials with probability $p_*$ before a predefined number $v$ of failures have occurred. Imagine the sequence of $A_j$ is extended to infinitely many Bernoulli trials: $A_1, \ldots, A_n, A'_{n+1}, \ldots$, where $\{A'_j\}_{j=n+1}^\infty$ are also i.i.d. Bernoulli with parameter $p_*$ and they are independent of $\{A_j\}_{j=1}^n$. Let $U$ be the number of success before $v$-th failure, then by definition, $U$ follows a negative binomial distribution $\mathrm{NB}(v, p_*)$. We can rewrite $U$ as a sum at a stopping time: $U \equiv M_{\tau'}$, where ${\tau' = \min\{t > 0}: t - M_t \geq v\}$.  By definition, $\tau^o \leq \tau'$ (a.s.), which indicates $M_{\tau^o} \leq M_{\tau'}$ because $M_t$ is nondecreasing with respect to $t$. Thus, we have proved that $M_{\tau^o} \preceq \mathrm{NB}(v, p_*)$.
\end{proof}

\begin{corollary} \label{col:nb}
Following the setting in Lemma~\ref{lm:basic_nb}, we consider the shrinking sum $\widetilde{M_t} = \sum_{j=1}^{n - t} A_j$ for $t = 0, 1, \ldots, n-1$. Let the filtration be $\widetilde{\mathcal{G}_t} = \sigma\left(\widetilde{M_t}, \{A_j\}_{j = n - t + 1}^{n}\right)$. Given a constant $v (\geq 1)$, we define a stopping time:
\begin{align} \label{eq:stop_back}
    \widetilde\tau  = \min\{0 \leq t < n: (n - t) - \widetilde{M_t} < v \text{ or } t = n - 1\},
\end{align}
then it still holds that $\widetilde{M_{\widetilde\tau}} \preceq \mathrm{NB}(v, p_*)$.
\end{corollary}
\begin{proof}
We first replace the notion of time $t$ by $n - s$, and let time runs backward: $s = n, n-1, \ldots, 1$. The above setting can be rewritten as $\widetilde{M_t} (= \sum_{j=1}^{n - t} A_j) \equiv M_{n-t} \equiv M_s $ and $\widetilde{\mathcal{G}_{t}} = \sigma\left(M_s, \{A_j\}_{j = s + 1}^{n}\right) =: \mathcal{G}_s^b$. Define a stopping time:
\begin{align} \label{eq:back_stop}
    \tau^b  = \max\{0 < s \leq n: s - M_s < v \text{ or } s = 1\},
\end{align}
which runs backward with respect to the filtration $\mathcal{G}_s^b$. By definition, we have $n - \widetilde\tau \equiv \tau^b$, and hence $\widetilde{M_{\widetilde\tau}} \equiv M_{\tau^b}$.

Now, we show that $M_{\tau^b} \equiv M_{\tau^o}$ for $\tau^o$ defined in Lemma~\ref{lm:basic_nb}. First, consider two edge cases: (1) if $t - M_t < v$ holds for every $0 < t \leq n$, then $\tau^b = n = \tau^o$, and thus $M_{\tau^b} = M_{\tau^o}$; (2) if $t - M_t \geq v$ holds for every $0 < t \leq n$, then $\tau^b = 1 = \tau^o$, and again $M_{\tau^b} = M_{\tau^o}$. Next, consider the case where $t - M_t < v$ for some $t$, and $t - M_t \geq v$ for some other $t$. Note that by definition, $\tau^b + 1$ is a stopping time with respect to $\mathcal{G}_{t}^o$, and $\tau^b + 1 = \tau^o$. Also, note that by the definition of $\tau^o$, we have $A_{\tau^o} = 0$, so $M_{\tau^o - 1} = M_{\tau^o}$. Thus, $M_{\tau^b} = M_{\tau^o - 1} = M_{\tau^o}$. Therefore, by Lemma~\ref{lm:basic_nb}, $\widetilde{M_{\widetilde\tau}} \equiv M_{\tau^b} \equiv M_{\tau^o} \preceq NB(v, p_*)$, as stated in the above Corollary.
\end{proof}

\begin{corollary} \label{col:nb_weighted}
Consider a weighted version of the setting in Corollary~\ref{col:nb}. Let the weights $\{W_j\}_{j=1}^n$ be a sequence of Bernoulli, such that (a)~$\sum_{j=1}^n W_j = m$ for a fixed constant $m \leq n$; and (b)~$A_j \mid \sigma\left(\{A_k, W_k\}_{k = j + 1}^n, W_j = 1\right)$ is a Bernoulli with parameter~$p_*$. Consider the sum $M_t^w = \sum_{j = 1}^{n-t} W_j A_j$. Given a constant $v (\geq 1)$, we define a stopping time:
\begin{align} \label{eq:stop_weight}
    \tau^w  =~& \min\{0 \leq t < n: \sum_{j = 1}^{n-t} W_j(1 - A_j) < v \text{ or } t = n - 1\}{}\\
    \equiv~& \min\{0 \leq t < n: \sum_{j = 1}^{n-t} W_j - M_t^w < v \text{ or } t = n - 1\} \nonumber,
\end{align}
then it still holds that $M_{\tau^w}^w \preceq \mathrm{NB}(v, p_*)$.
\end{corollary}
\begin{proof}
Intuitively, adding the binary weights should not change the distribution of the sum $M_{\tau^w}^w = \sum_{j = 1}^{n-\tau^w} W_j A_j$, since by condition~(b), $A_j$ is still a Bernoulli with parameter $p_*$ when it is counted in the sum. We formalize this idea as follows.

Let $\{B_l\}_{l=1}^m$ be a sequence of i.i.d. Bernoulli with parameter $p_*$, and denote the sum $\sum_{l=1}^{m-s} B_l$ as $\widetilde{M_s}(B)$. Let ${T(t) = m - \sum_{j=1}^{n - t} W_j}$, then the stopping time $\tau^w$ can be rewritten as
\begin{align} \label{eq:stop_image}
    \tau^w  \equiv \min\{0 \leq t < n: m - T(t) - \widetilde{M_{T(t)}}(B) < v \text{ or } t = n - 1\},
\end{align}
because $m - T(t) = \sum_{j = 1}^{n-t} W_j$ by definition, and
\begin{align} \label{eq:same_dist}
    \widetilde{M_{T(t)}}(B) = \sum_{l=1}^{m-T(t)} B_l\overset{d}{ = }  \sum_{j = 1}^{n-t} W_j A_j = M_t^w.
\end{align}
For simple notation, we present the reasoning of equation~\eqref{eq:same_dist} when $t = 0$ (for arbitrary $t$, consider the distributions conditional on $\{A_k, W_k\}_{k = n - t + 1}^n$). That is, we show that $\mathbb{P}(\sum_{l=1}^{m} B_l = x) =  \mathbb{P}(\sum_{j=1}^{n} W_j A_j = x)$ for every $x \geq 0$. Let $\{b_l\}_{j=1}^m \in \{0,1\}^m$, then we derive that 
\[
\mathbb{P}(\sum_{l=1}^{m} B_l = x) = \sum_{\sum b_l = x} \mathbb{P}(B_l = b_l \text{ for } l = 1, \ldots, n) = \sum_{\sum b_l = x} \prod_{l = 1}^m f^B(b_l),
\]
where $f^B$ is the probability mass function of a Bernoulli with parameter $p_*$. Let $\{a_k\}_{k=1}^{n-m} \in \{0,1\}^{n-m}$, then for the weighted sum,
\begin{align*}
    &\mathbb{P}(\sum_{j=1}^{n} W_j A_j = x){}\\
    =~& \sum_{\sum b_l = x} \sum_{\sum w_j = m} \sum_{a_k} \mathbb{P}(A_j = b_l \text{ if } w_j = 1; A_j = a_k \text{ if } w_j = 0; W_j = w_j \text{ for } i = 1,\ldots, n){}\\
    =~& \sum_{\sum b_l = x} \prod_{l = 1}^m f^B(b_l) \underbrace{\sum_{\sum w_j = m} \sum_{\sum a_k}  \prod_{w_j = 0} \mathbb{P}(A_j = a_k \mid \sigma\left(\{A_k, W_k\}_{k = j + 1}^n, W_j = 0\right) \prod_{j=1}^n \mathbb{P}(W_j = w_j \mid \{A_k, W_k\}_{k = j + 1}^n)}_{C \quad (\text{a constant with respect to } x)}{}\\
    =~& C\sum_{\sum b_l = x} \prod_{l = 1}^m f^B(b_l) = C\mathbb{P}(\sum_{l=1}^{m} B_l = x),
\end{align*}
for every possible value $x \geq 0$, which implies that $\mathbb{P}(\sum_{l=1}^{m} B_l = x)$ and $\mathbb{P}(\sum_{j=1}^{n} W_j A_j = x)$ have the same value; and hence we conclude equation~\eqref{eq:same_dist}. It follows that the filtration for both the stopping time~$\tau^w$ and the sum $M^w_{t^w}$, denoted as $\sigma\left(\sum_{j=1}^{n-t}W_j, M^w_{t^w}, \{A_j, W_j\}_{j = n - t + 1}^n \right)$, has the same probability measure as ${\sigma\left(m - T(t), \widetilde{M_{T(t)}}(B), \{A_j, W_j\}_{j = n - t + 1}^n\right)}$. Thus, the sums at the stopping time have the same distribution, ${M_{\tau^w}^w \overset{d}{=} \widetilde{M_{T(\tau^w)}}(B)}$.
The proof completes if $\widetilde{M_{T(\tau^w)}}(B) \preceq \mathrm{NB}(v, p_*)$. It can be proved once noticing that stopping rule~\eqref{eq:stop_image} is similar to stopping rule~\eqref{eq:stop_back} except $T(t)$ is random because of $W_j$, so we can condition on $\{W_j\}_{j=1}^n$ and apply Corollary~\ref{col:nb}; and this concludes the proof.
\end{proof}

\begin{corollary} \label{col:nb_cons}
In Corollary~\ref{col:nb_weighted}, consider $A_j$ with different parameters. Suppose $A_j \mid \sigma\left(\{A_k, W_k\}_{k = j + 1}^n, W_j = 1\right)$ is a Bernoulli with parameter $p\left(\{A_k, W_k\}_{k = j + 1}^n\right)$ for every ${j = 1, \ldots, n}$. Given a constant $p_* \in (0, 1)$, if the parameters satisfy that $p\left(\{A_k, W_k\}_{k = j + 1}^n\right) \leq p_*$ for all~$j = 1, \ldots, n$, then it still holds that $M_{\tau^w}^w \preceq \mathrm{NB}(v, p_*)$.
\end{corollary}
\begin{proof}
We first construct Bernoulli with parameter $p_*$ based on $A_j$ by an iterative process. Start with $j = n$. Let $C_n$ be a Bernoulli independent of $\{A_k\}_{k=1}^n$ with parameter $\frac{p_* - p_n}{1 - p_n}$, where $p_n = \mathbb{E}(A_n \mid W_n = 1).$ Construct 
\begin{align}
    B_n = A_n\one\left(A_n = 1\right) + C_n\one\left(A_n = 0\right),
\end{align}
which thus satisfies that $\mathbb{E}(B_n \mid W_n = 1) = p_*$, and that $B_n \geq A_n$ (a.s.). Now, let $j = j - 1$ where we consider the previous random variable. Let $C_j$ be a Bernoulli independent of $\{A_k\}_{k=1}^{j}$, with parameter 
\begin{align} \label{eq:para_c}
    \frac{p_* - \widetilde p\left(\{B_k, W_k\}_{k = j + 1}^n\right)}{1 - \widetilde p\left(\{B_k, W_k\}_{k = j + 1}^n\right)},
\end{align}
where $\widetilde p\left(\{B_k, W_k\}_{k = j + 1}^n\right) = \mathbb{E}\left[A_j \mid \sigma\left(\{B_k, W_k\}_{k = j + 1}^n, W_j = 1\right)\right]$ (note that the parameter for $C_j$ is well-defined since $\widetilde p\left(\{B_k, W_k\}_{k = j + 1}^n\right) \leq p_*$ by considering the expectation further conditioning on $\{A_k\}_{k = j + 1}^n$). Then, we construct $B_j$ as
\begin{align} \label{eq:construct_b}
    B_j = A_j\one\left(A_j = 1\right) + C_j\one\left(A_j = 0\right),
\end{align}
which thus satisfies that $\mathbb{E}\left[B_j \mid \sigma\left(\{B_k, W_k\}_{k = j + 1}^n, W_j = 1\right)\right] = p_*$, and that $B_j \geq A_j$ (a.s.). 

Now, consider two procedures for $\{A_j\}_{j=1}^n$ and $\{B_j\}_{j=1}^n$ with the same stopping rule~\eqref{eq:stop_weight} in Corollary~\ref{col:nb_weighted}, where the sum of~$A_j$ is denoted as $M_t^w(A)$ and the stopping time as $\tau^w_A$ (and the similar notation for $B_j$).
Since construction~\eqref{eq:construct_b} ensures that $B_j \geq A_j$ for every $j = 1, \ldots, n$, we have $M_t^w(B) \geq M_t^w(A)$ for every $t$; and hence, $\tau^w_A \geq \tau^w_B$. It follows that
\[
M_{\tau^w_A}^w(A) \leq M_{\tau^w_B}^w(A) \leq 
M_{\tau^w_B}^w(B) \preceq \mathrm{NB}(v, p_*),
\]
where the first inequality is because $M_t^w$ is nonincreasing with respect to $t$, and the last step is the conclusion of Corollary~\ref{col:nb_weighted}; this completes the proof.
\end{proof}

\subsection{Proof of Theorem~\ref{thm:fwer}.}
\begin{proof}
We discuss three cases: (1) the simplest case where all the hypotheses are null, and the null $p$-values are uniformly distributed; (2)~the case where non-nulls may exist, and the null $p$-values are uniformly distributed; and finally (3)~the case where non-nulls may exist, and the null $p$-values can be mirror-conservative.

\paragraph{Case 1: nulls only and null $p$-values uniform.}
By Lemma~\ref{lm:seq}, $\{\one\left(h\left(P_{\pi_j}\right) = 1\right)\}_{j=1}^n$ are i.i.d. Bernoulli with parameter~$p_*$. Observe that the stopping rule in Algorithm~\ref{alg:apt}, $\widehat{\text{FWER}_t} \equiv 1 - (1 - p_*)^{|\mathcal{R}_t^-| + 1} \leq \alpha$, can be rewritten as $|\mathcal{R}_t^-| + 1 \leq v$ where
\begin{align} \label{eq:v_def}
    v = \left\lfloor\frac{\log (1 - \alpha)}{\log (1 - p_*)} \right\rfloor,
\end{align}
which is also equivalent as $|\mathcal{R}_t^-| < v$. We show that the number of false rejections is stochastically dominated by $\mathrm{NB}(v, p_*)$ by Corollary~\ref{col:nb}. Let ${A_j = \one\left(h\left(P_{\pi_j}\right) = 1\right)}$ and $\widetilde{M_{t}} = \sum_{j = 1}^{n - t} \one\left(h\left(P_{\pi_j}\right) = 1\right)$. The stopping time is  
$\widetilde\tau = \min\{0 \leq t < n: |\mathcal{R}_t^-| = (n - t) - \widetilde{M_{t}} < v \text{ or } t = n -1\}$. The number of rejections at the stopping time is 
\[
|\mathcal{R}_{\widetilde\tau}^+| \equiv \sum_{j = 1}^{n -\widetilde\tau} \one\left(h\left(P_{\pi_j}\right) = 1\right) \equiv \widetilde{M_{\widetilde\tau}} \preceq \mathrm{NB}(v, p_*),
\]
where the last step is the conclusion of Corollary~\ref{col:nb}. Note that we assume all the hypotheses are null, so the number of false rejections is $|\mathcal{R}_{\widetilde\tau}^+ \cap \mathcal{H}_0| = |\mathcal{R}_{\widetilde\tau}^+| \preceq \mathrm{NB}(v, p_*)$. Thus, FWER is upper bounded:
\begin{align} \label{eq:fwer_control}
    \mathbb{P}(|\mathcal{R}_{\widetilde\tau}^+ \cap \mathcal{H}_0| \geq 1) \leq 1 - (1 - p_*)^v \leq \alpha,
\end{align}
where the last inequality follows by the definition of $v$ in~\eqref{eq:v_def}. Thus, we have proved FWER control in Case~1.

\textbf{Remark:} This argument also provides some intuition on the FWER estimator~\eqref{eq:fwer_hat}: $\widehat{\text{FWER}_t} = 1 - (1 - p_*)^{|\mathcal{R}_t^-| + 1}$. Imagine we run the algorithm for one time without any stopping rule until time $t_0$ to get an instance of $\widehat{\text{FWER}_{t_0}}$, then we run the algorithm on another independent dataset, which stops once $\widehat{\text{FWER}_t}\leq \widehat{\text{FWER}_{t_0}}$. Then in the second run, FWER is controlled at level $ \widehat{\text{FWER}_{t_0}}$.

\paragraph{Case 2: non-nulls may exist and null $p$-values are uniform.}
We again argue that the number of false rejections is stochastically dominated by $\mathrm{NB}(v, p_*)$, and in this case we use Corollary~\ref{col:nb_weighted}. Consider $A_j = \one\left(h\left(P_{\pi_j}\right) = 1\right)$ and ${W_j = \one\left(\pi_j \in \mathcal{H}_0\right)}$, which satisfies condition~(b) in Corollary~\ref{col:nb_weighted} according to Corollary~\ref{col:seq}. Let $m = |\mathcal{H}_0|$, then ${\sum_{j=1}^n W_j = m}$, which corresponds to condition~(a). Imagine an algorithm stops once
\begin{align} \label{eq:stop_imagine}
    \sum_{j = 1}^{n - t} \one\left(h\left(P_{\pi_j} \right) = - 1 \cap \pi_j \in \mathcal{H}_0\right)
    = \sum_{j = 1}^{n - t} W_j (1 - A_j) < v,
\end{align}
and we denote the stopping time as $\tau^w$. By Corollary~\ref{col:nb_weighted}, the number of false rejections in this imaginary case is
\[
\sum_{j = 1}^{n - \tau^w} \one\left(h\left(P_{\pi_j} \right) = 1 \cap \pi_j \in \mathcal{H}_0\right) = \sum_{j = 1}^{n - t} W_j  A_j
     = M_{\tau^w}^w \preceq \mathrm{NB}(v, p_*).
\]
Now, consider the actual \apt which stops when $|R_t^-| = (n - t) - \sum_{j = 1}^{n -t} \one\left(h\left(P_{\pi_j}\right) = 1\right) < v$, and denote the true stopping time as~$\tau_T^w$. Notice that at the stopping time, it holds that
\begin{align*}
    &\sum_{j = 1}^{n - \tau_T^w} \one\left(h\left(P_{\pi_j} \right) = - 1 \cap \pi_j \in \mathcal{H}_0\right){}\\ 
    \leq~& \sum_{j = 1}^{n - \tau_T^w} \one\left(h\left(P_{\pi_j} \right) = - 1\right){}\\ 
    =~& (n - \tau_T^w) - \sum_{j = 1}^{n - \tau_T^w} \one\left(h\left(P_{\pi_j}\right) = 1\right)< v,
\end{align*}
which means that stopping rule~\eqref{eq:stop_imagine} is satisfied at $\tau_T^w$. Thus, $\tau_T^w \geq \tau^w$ and $M_{\tau_T^w}^w \leq M_{\tau^w}^w$ (because $M_t^w$ is nonincreasing with respect to $t$). It follows that the number of false rejections is
\[
|\mathcal{R}_{\tau_C^w}^+ \cap \mathcal{H}_0| \equiv \sum_{j = 1}^{n - \tau_C^w} \one\left(h\left(P_{\pi_j}\right) = 1 \cap \pi_j \in \mathcal{H}_0\right) \equiv M_{\tau_C^w}^w \leq  M_{\tau^w}^w \preceq \mathrm{NB}(v, p_*).
\]
We then prove FWER control using a similar argument as~\eqref{eq:fwer_control}:
\[
\mathbb{P}(|\mathcal{R}_{\tau^w}^+ \cap \mathcal{H}_0| \geq 1) \leq 1 - (1 - p_*)^v \leq \alpha,
\]
which completes the proof of Case~2.
 
\paragraph{Case 3: non-nulls may exist and null $p$-values can be mirror-conservative.}
In this case, we follow the proof of Case~2 except additionally conditioning on all the masked $p$-values, $\{g(P_{\pi_k})\}_{k=1}^n$. By Corollary~\ref{col:seq_cons} and Corollary~\ref{col:nb_cons}, we again conclude that the number of false rejections is dominated by a negative binomial:
\[
|\mathcal{R}_{\tau_C^w}^+ \cap \mathcal{H}_0| \preceq \mathrm{NB}(v, p_*),
\]
if given $\{g(P_{\pi_k})\}_{k=1}^n$. Thus, FWER conditional on $\{g(P_{\pi_k})\}_{k=1}^n$ is upper bounded:
\[
\mathbb{P}\left(|\mathcal{R}_{\tau^w}^+ \cap \mathcal{H}_0| \geq 1 \middle| \{g(P_{\pi_k})\}_{k=1}^n\right) \leq 1 - (1 - p_*)^v \leq \alpha,
\]
which implies the FWER control by the law of iterated expectations. This completes the proof of Theorem~\ref{thm:fwer}.
\end{proof}

\section{An alternative perspective: closed testing} \label{apd:closed_testing}
This section summarizes the comments from Jelle Goeman, who kindly points out the connection between our proposed method and the \textit{closed testing} \citep{marcus1976closed}. Closed testing is a general framework that generates a procedure with FWER control given any test with Type 1 error control. Specifically, we reject $H_i$ if all possible sets of hypotheses involving~$H_i$, denoted as $U \ni i$, can be rejected by a ``local'' test for hypotheses in $U$ with Type 1 error control at level $\alpha$. 

The \apt we propose shares some commonalities with the \textit{fallback procedure} \citep{wiens2005fallback}, which can be viewed as a shortcut of a closed testing procedure. We briefly describe the commonalities and differences next. Let $v$ be a prespecified positive integer. The fallback procedure orders the hypotheses from most to least interesting, and proceeds to test them one by one at level $\alpha/v$ until it has failed to reject $v$ hypotheses. The hypothesis ordering is allowed to be data-dependent as long as the ordering is independent of the $p$-values, corresponding to ordering by the side information $x_i$ in our language. This procedure is essentially also what the \apt does except~(a) the \apt uses the {\v{S}}id{\'a}k correction instead of the Bonferroni correction;~(b) we are interested in whether rejecting each hypothesis instead of adjusting individual $p$-values, so the ordering only needs to be independent of reject/non-reject status instead of on the full $p$-values, which allows us to split each $p$-value into $h(P_i)$ and $g(P_i)$;~(c) under the assumption of independent null $p$-values, we are allowed to use the $p$-values excluded from the candidate rejection set $\mathcal{R}_t$ as independent information to create the ordering. The latter two differences enable the \apt to be interactive based on a considerably large amount of data information. 

\subsection{Alternative proof of Theorem~\ref{thm:fwer}}
The above observation leads to a simple proof of the error control guarantee without involving any martingales or negative binomial distributions, once we rewrite the \apt in the language of closed testing. 
\begin{proof}
For simplicity, we consider the nulls with only uniform $p$-values. Let $v$ be a prespecified positive integer, and define $p_* = 1 - (1 - \alpha)^{1/v}$. Imagine that the \apt does not have a stopping rule and let $\pi_n, \ldots, \pi_1$ be the order in which the hypotheses are chosen by an analyst, where each choice $\pi_t$ can base on all the information in $\mathcal{F}_{n - t}$. 

Here, we construct a closed testing procedure by defining a local test with Type 1 error control for an arbitrary subset~$U \in [n]$ of size $|U|$. Sort the hypotheses in $U$ according to the analyst-specified ordering from the last $\pi_n$ to the first chosen $\pi_1$. If the number of hypotheses in $U$ is larger than $v$, define $U_v$ as the subset of $U$ of size $v$ corresponding to the hypotheses in $U$ that are chosen last. For example, if $U = [n]$, we have $U_v = \{\pi_v, \ldots, \pi_1\}$. If $|U| \leq v$, define $U_v = U$. We reject the subset $U$ if $h(P_i) = 1$ (i.e., $P_i \leq p_*)$ for at least one $i \in U_v$. This is a valid local test, since it controls the Type 1 error when all the hypotheses in $U$ are null. To verify the error control, notice that $h(P_i)$'s are independent and follows $\text{Bernoulli}(p_*)$, and $U_v$ is independent of $\{h(P_i)\}_{i \in U_v}$ by the construction of sequence $\pi_1, . . . , \pi_n$, so the Type 1 error satisfy
\begin{align*}
    \mathbb{P}(\exists i \in U_v: h(P_i) = 1) \leq 1 - (1 - p_*)^v,
\end{align*}
which is less than~$\alpha$ by the definition of $v$ and $p_*$. Indeed, the local test corresponds to a {\v{S}}id{\'a}k correction for $v$ number of hypotheses. Through closed testing, this local test leads to a valid test with FWER control. 

Next, we show that the rejection set from the \apt, $\mathcal{R}_\tau^+$, is included in the rejection set from the above closed testing procedure. Choose any hypothesis $j \in \mathcal{R}_\tau^+$ and any set $W \ni j$. If $H_j$ is among the last $v$ hypotheses last chosen in $W$ (or if $|W| \leq v$), the local test for $W$ reject the null since $P_j \leq p_*$ by the definition of $\mathcal{R}_\tau^+$. Otherwise, the $v$ hypotheses last chosen in $W$ are all chosen after $H_j$. Since $j \in \mathcal{R}_\tau^+$ and by the definition of $\tau$, we have $|\mathcal{R}_\tau^-| \leq v - 1$. That is, there can be at most $v - 1$ hypotheses among these $v$ such that $h(P_i) = -1$, so set~$W$ is rejected by the local test as described in the previous paragraph. It follows from the definition of FWER and the error control of the larger (or equivalent) rejection set from the closed testing procedure that $\mathcal{R}_\tau^+$ has FWER control.
\end{proof}

\subsection{Improvement on an edge case}
From the closed testing procedure constructed in the above proof, we observe that the local tests do not exhaust the $\alpha$-level for intersections of less than $v$ hypotheses. This suboptimality can be remedied, but it will only improve power for rejecting all hypotheses given that almost all are already rejected (i.e., most subsets~$U$ with $|U| > v$ are rejected by the local test).  In the \apt, such a case potentially corresponds to the case where the initial rejection set has less than $v$ hypotheses with negative $h(P_i)$, so the algorithm stops before shrinking $\mathcal{R}_0$, and reject all the hypotheses with positive $h(P_i)$. However, we might not fully use the error budget because $\widehat{\text{FWER}_0} < \alpha$. However, we might not fully use the error budget because $\widehat{\text{FWER}_0} < \alpha$. To improve power and efficiently use all the error budget, we propose randomly rejecting the hypotheses with a negative $h(P_i)$ if the algorithm stops at step $0$.

\begin{algorithm}[h!]
   \caption{The adjusted \apt}
   \label{alg:adj_apt}
\begin{algorithmic}
   \STATE {\bfseries Input:} Side information and $p$-values $\{x_i, P_i\}_{i = 1}^n$, target FWER level~$\alpha$, and parameter~$p_*$;
   \STATE {\bfseries Procedure:} 
   \STATE Initialize $\mathcal{R}_0 = [n]$;
   \IF{$\widehat{\text{FWER}}_0 \equiv 1 - (1 - p_*)^{|\mathcal{R}_0^-| + 1} \leq \alpha$}
   \STATE Obtain $n$ independent indicators from a Bernoulli distribution with probability $1 - (1 - \alpha + \widehat{\text{FWER}_0})^{1/|\mathcal{R}_0^-|}$, denoted as $\{I_i\}_{i\in [n]}$;
   \STATE Reject $\{H_i: i \in [n], h(P_i) = 1 \text{ or } I_i = 1\}$ and exit;
   \ELSE
   \FOR{$t=1$ {\bfseries to} $n$}
   \STATE 1.~Pick any $i_t^* \in \mathcal{R}_{t-1}$, using $\{x_i, g(P_i)\}_{i=1}^n$ and $\{h(P_i)\}_{i \notin \mathcal{R}_{t-1}}$;
   \STATE 2.~Exclude $i_t^*$ and update ${\mathcal{R}_t = \mathcal{R}_{t-1}\backslash \{i_t^*\}}$; 
   \IF{$\widehat{\text{FWER}}_t \equiv 1 - (1 - p_*)^{|\mathcal{R}_t^-| + 1} \leq \alpha$}
   \STATE Reject $\{H_i: i \in \mathcal{R}_t, h(P_i) = 1\}$ and exit;
   \ENDIF
   \ENDFOR
   \ENDIF
\end{algorithmic}
\end{algorithm}

Recall that the number of negative $h(P_i)$ is $|\mathcal{R}_0^-|$. For each hypothesis with a negative $h(P_i)$, we independently decide to reject it with probability $1 - (1 - \alpha_\text{re})^{1/|\mathcal{R}_0^-|}$, where $\alpha_\text{re} := \alpha - \widehat{\text{FWER}_0}$ denotes the remaining error budget after rejecting all the hypotheses with positive $h(P_i)$'s. We summarize the adjusted \apt in Algorithm~\ref{alg:adj_apt}. To see the error control guarantee of this improved algorithm, notice that
\begin{align*}
    &\mathbb{P}(\exists i \in \mathcal{H}_0: H_i \text{ is rejected}){}\\ 
    \leq~& \mathbb{P}(\exists i \in \mathcal{H}_0: h(P_i) = 1) + \mathbb{P}(\exists i \in \mathcal{H}_0: h(P_i) = -1 \text{ and } H_i \text{ is rejected}){}\\
    \leq~& \widehat{\text{FWER}_0} + \mathbb{P}(\exists i \in \mathcal{R}_0^-: H_i \text{ is rejected}){}\\
    \leq~& \widehat{\text{FWER}_0} + \alpha_\text{re} = \alpha,
\end{align*}
where $\mathbb{P}(\exists i \in \mathcal{H}_0: h(P_i) = 1) \leq \widehat{\text{FWER}_0}$ follows the argument using negative binomial distribution as in the proof of the original algorithm; and $\mathbb{P}(\exists i \in \mathcal{R}_0^-: H_i \text{ is rejected}) \leq \alpha_\text{re}$ is the result of a {\v{S}}id{\'a}k correction. 



\section{Sensitivity analysis}
The \apt is proved to have valid error control when the nulls are mutually independent and independent of the non-nulls. In this section, we evaluate the performance of the \apt under correlated $p$-values. Our numerical experiments construct a grid of hypotheses as described in the setting in Section~\ref{set:cluster}. The p-values are generated as
\begin{align} \label{eq:p_dep}
    P_i = 1 - \Phi(Z_i), \text{where } Z = (Z_1, \ldots, Z_n) \sim N(\mu, \Sigma),
\end{align}
where $\mu = 0$ for the nulls and $\mu = 3$ for the non-nulls. The covariance matrix $\Sigma$, which is identity matrix in the main paper, is now set to an equi-correlated matrix:
\begin{align} \label{eq:rho_mat}
    \begin{bmatrix}
1 & \rho & \cdots & \rho\\
\rho & 1 & \cdots & \rho\\
\vdots & \vdots & \vdots & \vdots\\
\rho & \rho & \cdots & 1
\end{bmatrix}.
\end{align}

\begin{figure*}[h!]
    \centering
    \begin{subfigure}[t]{0.24\textwidth}
        \centering
        \includegraphics[width=0.8\linewidth]{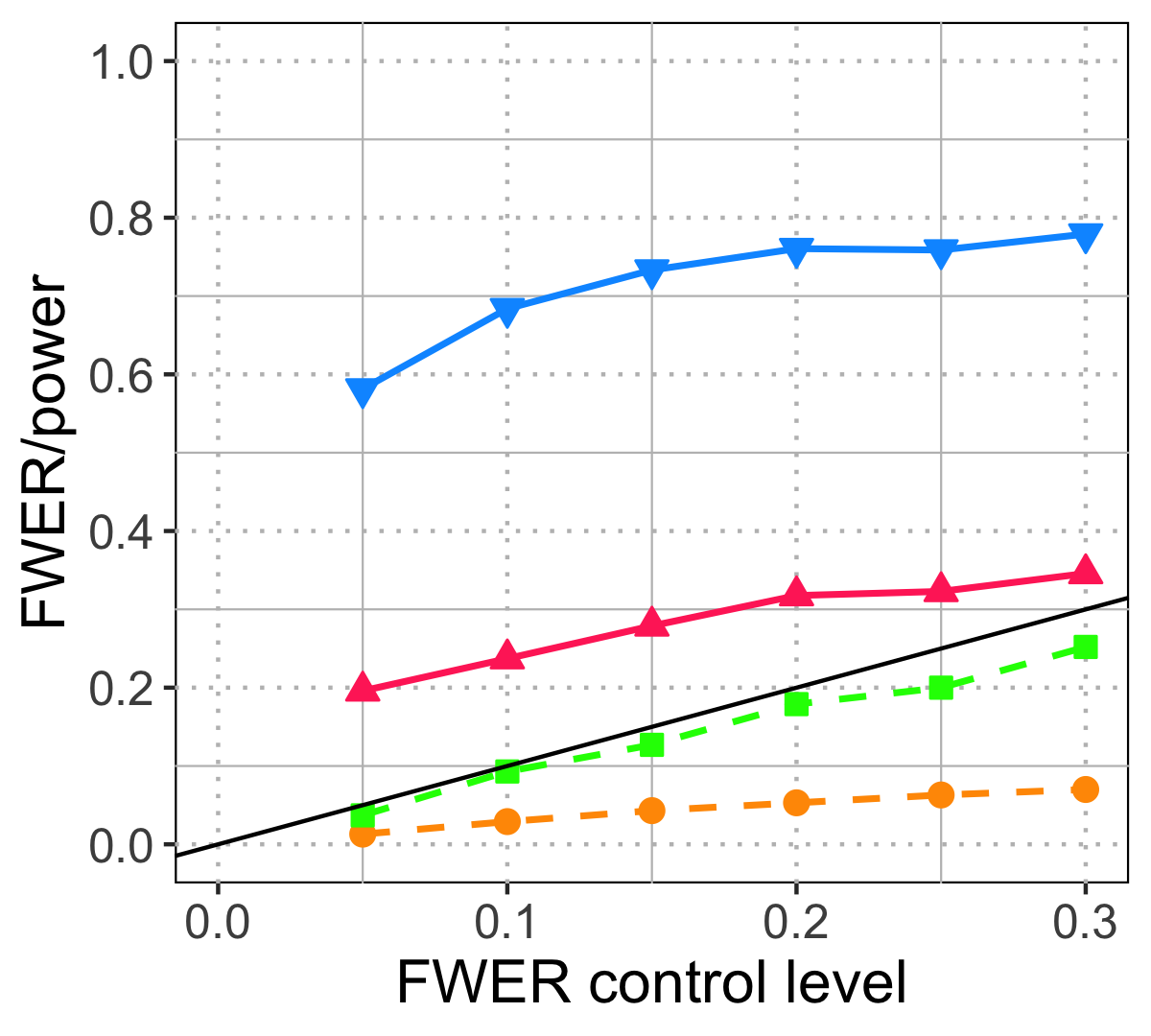}
        \caption{Positively correlated case where $\rho = 0.5$ in the covariance matrix~\eqref{eq:rho_mat}. The non-null mean value is $3$.}
        \label{fig:dep_pos}
    \end{subfigure}
    \hfill
    \begin{subfigure}[t]{0.24\textwidth}
        \centering
        \includegraphics[width=0.8\linewidth]{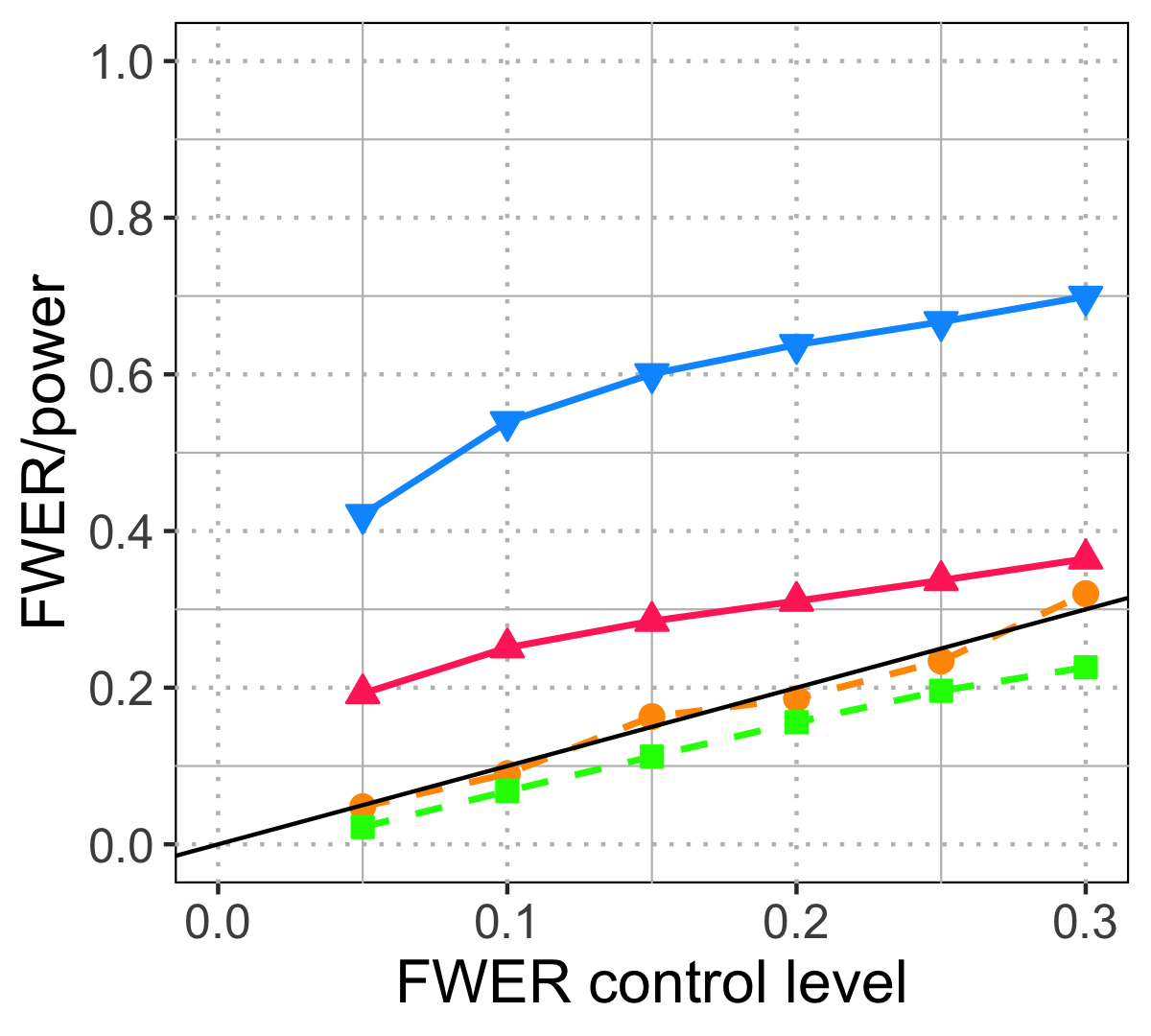}
        \caption{Negatively correlated case where $\rho = -0.5/n$ in the covariance matrix~\eqref{eq:rho_mat}. Non-null mean value is $3$.}
        \label{fig:dep_neg}
    \end{subfigure}
    \hfill
    \begin{subfigure}[t]{0.24\textwidth}
        \centering
        \includegraphics[width=0.8\linewidth]{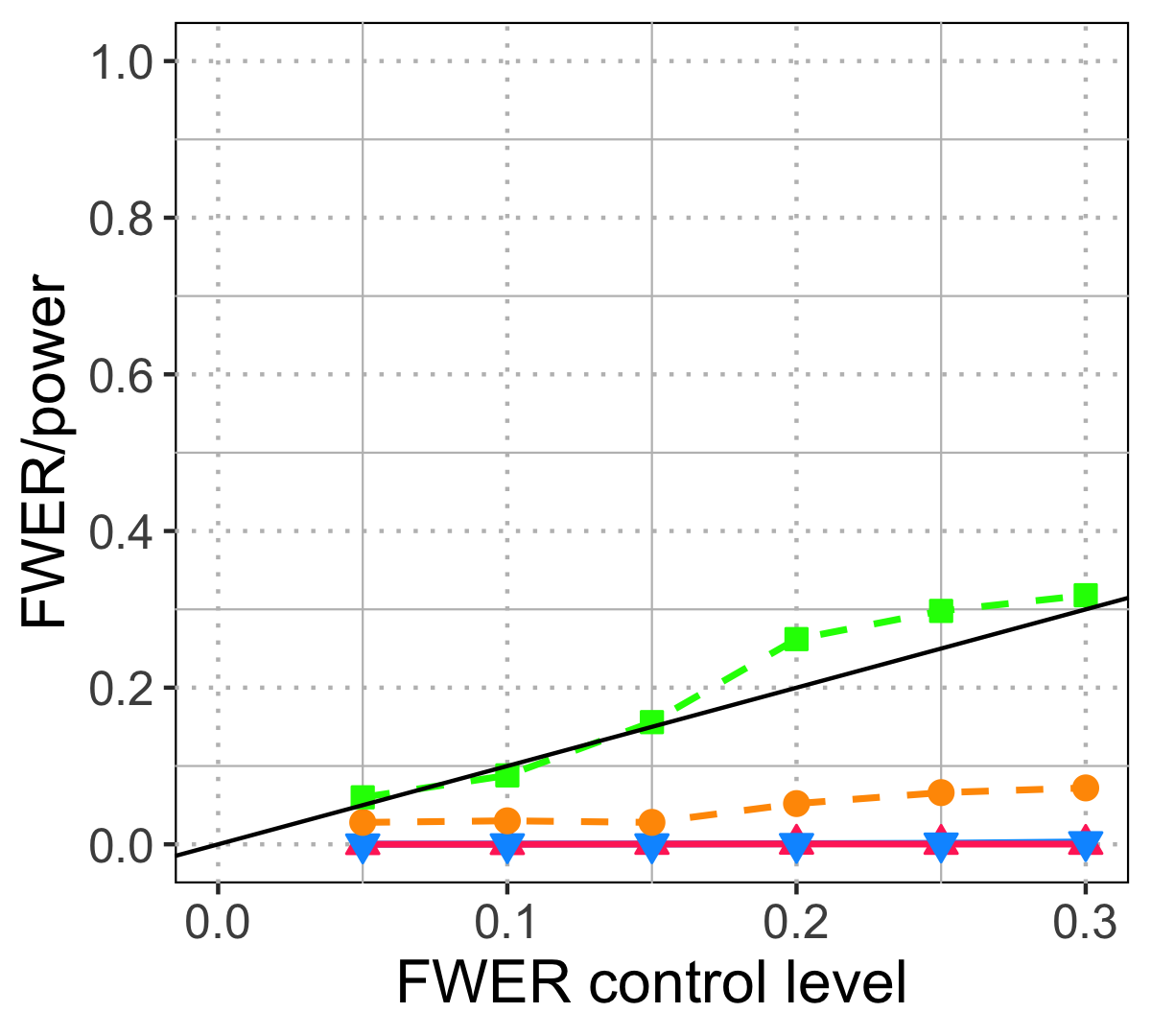}
        \caption{Negatively correlated case where $\rho = -0.5/n$ in the covariance matrix~\eqref{eq:rho_mat}. All hypotheses are nulls.}
        \label{fig:dep_pos_0}
    \end{subfigure}
    \hfill
    \begin{subfigure}[t]{0.24\textwidth}
        \centering
        \includegraphics[width=0.8\linewidth]{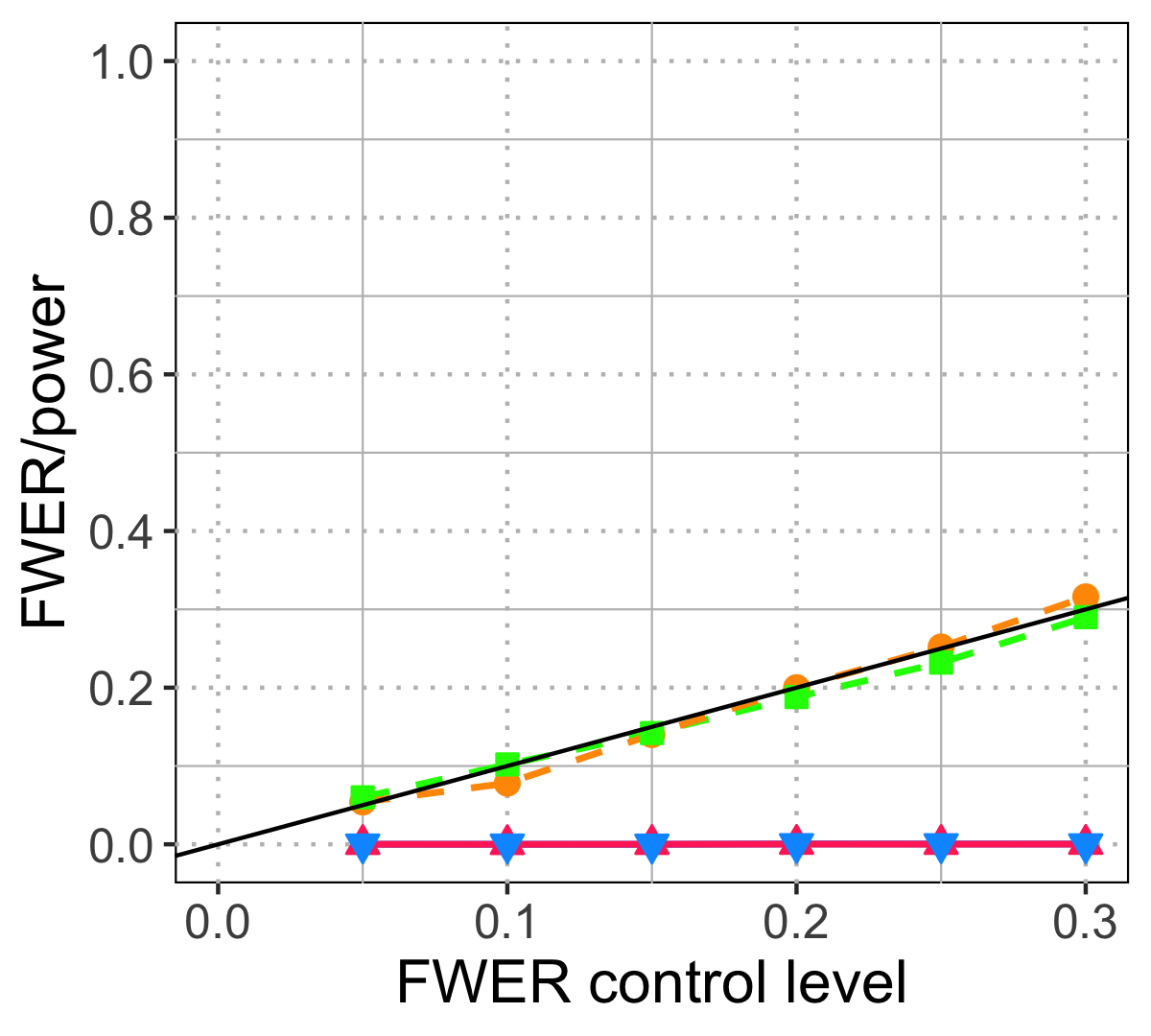}
        \caption{Negatively correlated case where $\rho = -0.5/n$ in the covariance matrix~\eqref{eq:rho_mat}. All hypotheses are nulls.}
        \label{fig:dep_neg_0}
    \end{subfigure}
\vfill
   \begin{subfigure}[t]{1\textwidth}
        \centering
        \includegraphics[width=0.6\linewidth]{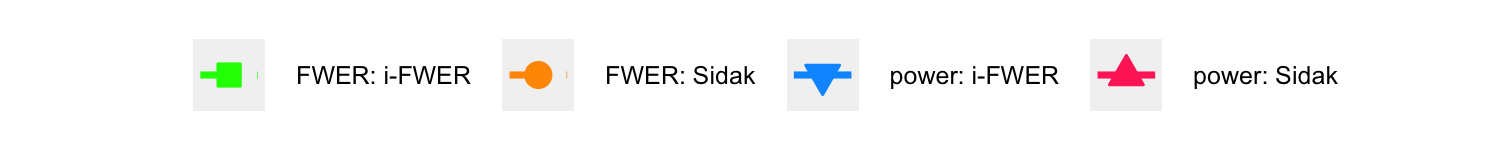}
    \end{subfigure}
    \caption{FWER and power of the \apt and the {\v{S}}id{\'a}k correction for dependent p-values generated by Gaussians as in~\eqref{eq:p_dep} with covariance matrix~\eqref{eq:rho_mat} when the targeted level of FWER control varies in $(0.05, 0.1, 0.15, 0.2, 0.25, 0.3)$. The \apt appears to control FWER below the targeted level and has relatively high power.}
    \label{fig:depend}
\end{figure*}

Under both the positively correlated case ($\rho = 0.5$) and the negatively correlated case ($-\rho = 0.5/n$ to guarantee that $\Sigma$ is positive semi-definite), the \apt seems to maintain the FWER control at most target levels even when all the hypotheses are nulls (see Figure~\ref{fig:dep_pos_0} and~\ref{fig:dep_neg_0}), and has higher power than the {\v{S}}id{\'a}k correction (see Figure~\ref{fig:dep_pos} and~\ref{fig:dep_neg}).

\section{More results on the application to genetic data}
Section~\ref{sec:real_data} presents the number of rejections of the \apt when the masking uses the tent function. We evaluate the \apt when using the other three masking functions under the same experiments, but for simplicity, we only present the result when the FWER control is at level $\alpha = 0.2$ (see Table~\ref{tab:airway_add}). Overall, the gap function leads to a similar number of rejections as the tent function, consistent with the numerical experiments. However, the railway (gap-railway) function leads to fewer rejections than the tent (gap) function, which seems counterintuitive. Upon a closer look at the $p$-values, we find that the null $p$-values are not uniform or have an increasing density (see Figure~\ref{fig:hist_airway}). As a result, when using the tent function, there are fewer masked $p$-values from the nulls that could be confused with those of the non-nulls (with huge $p$-values), compared with using the railway function where the masked $p$-values of the confused nulls are those close to the masking parameter (around $0.02$).

\begin{table}[ht]
\caption{Number of rejections by \apt using different masking functions when $\alpha = 0.2$. The tent function and the gap function leads to more rejections compared with the railway function and the gap-railway function. The parameters in the gap and gap-railway function are set to $p_l = p_*$ and $p_u = 0.5$, and we need $p_l < \alpha/2$ for the test to make any rejection under level $\alpha$.}
\label{tab:airway_add}
\centering
\begin{tabular}{|c|c|c|c|}
\hline
Masing function & $p_* = \alpha/2$ & $p_* = \alpha/10$ & $p_* = \alpha/20$ \\\hline
Tent & 1752 & 1848 & 1794\\ \hline
Railway & 1778 & 1463 & 1425\\ \hline
Gap & NA & 1802 & 1846 \\ \hline
Gap-railway & NA & 1764 & 1788\\ \hline
\end{tabular}
\end{table}

\begin{figure*}[ht]
    \centering
        \includegraphics[width=0.2\linewidth]{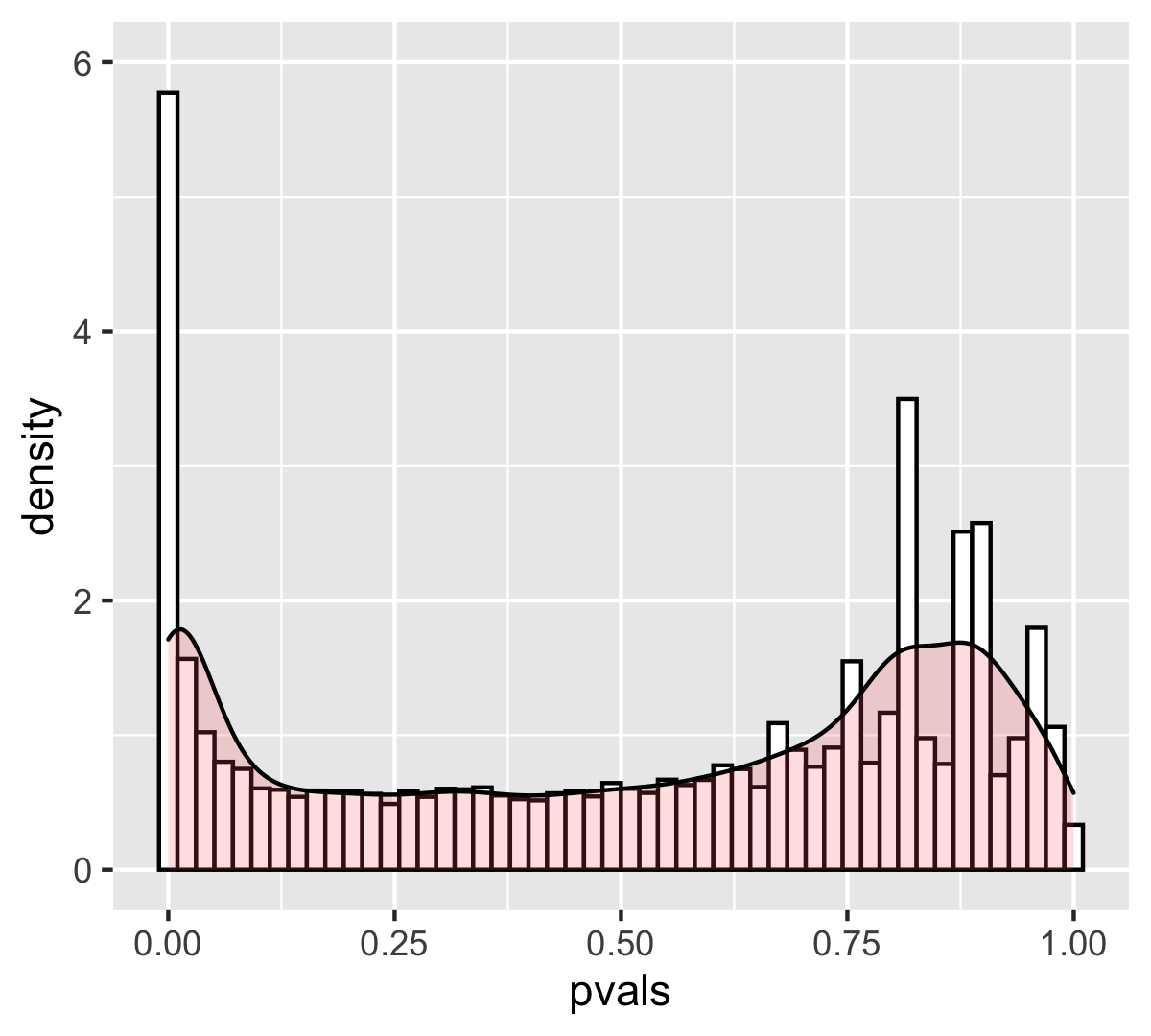}
        \caption{Histogram of $p$-values in the airway dataset. The number of $p$-values that are close to one is less than those that are close to the cutting point of the masking function (say $0.02$). Consequently, the tent (gap) function leads to more rejections than the railway (gap-railway) function.}
    \label{fig:hist_airway}
\end{figure*}

\section{Error control for other masking functions} \label{apd:error_mask}

The proof in Appendix~\ref{apd:thm1} is for the \apt with the original tent masking function. In this section, we check the error control for two new masking functions introduced in Section~\ref{sec::mask}.

\subsection{The railway function} 
We show that the \apt with the ``railway'' function~\eqref{eq:mask_railway} has FWER control, if the null $p$-values have non-decreasing densities. We again assume the same independence structure as in Theorem~\ref{thm:fwer} that the null $p$-values are mutually independent and independent of the non-nulls.  

The proof in Appendix~\ref{apd:thm1} implies that under the same independence assumption, the FWER control is valid if the null $p$-values satisfy condition~\eqref{eq:cond_cons}. When using the railway masking function, condition~\eqref{eq:cond_cons} is indeed satisfied if the null $p$-values have nondecreasing $f$ since
\begin{align*}
    \mathbb{P}(h(P) = 1 \mid g(P) = a)
    =~& \frac{p_* f(a)}{p_* f(a) + (1 - p_*)f(\frac{1 - p_*}{p_*}a + p_*)} {}\\
    =~& \frac{p_*}{p_* + (1 - p_*)f(\frac{1 - p_*}{p_*}a + p_*)/ f(a)} {}\\
    \leq~& p_*,
\end{align*}
for every $a \in (0, p_*)$. Then, we can prove the FWER control following the same argument as Appendix~\ref{apd:thm1}.

\subsection{The gap function}
The essential difference of using the gap function instead of the tent function is that here, $\one\left(h(P) = 1\right)$ for the nulls follow a Bernoulli distribution with a different parameter, $\widetilde{p} = \mathbb{P}(h(P) = 1 \mid P < p_l \text{ or } P > p_u) = \frac{p_l}{p_l + 1 - p_u}$. Specifically, we replace replace condition~\eqref{eq:cond_cons} by  
\begin{align} \label{eq:cond_gap}
    \mathbb{P}(h(P) = 1 \mid g(P) = a) \leq  \widetilde{p},
\end{align} 
for every $a \in (0, p_l)$, which holds for $p$-values with non-decreasing densities because
\begin{align*}
    \mathbb{P}(h(P) = 1 \mid g(P) = a)
    =~& \frac{p_l f(a)}{p_l f(a) + (1 - p_u)f(1 - \frac{1 - p_u}{p_l}a)} {}\\
    =~& \frac{p_l}{p_l + (1 - p_u)f(1 - \frac{1 - p_u}{p_l}a)/ f(a)} {}\\
    \leq~& \widetilde{p}.
\end{align*}

We also replace all $p_*$ by $\widetilde{p}$ and get a the new FWER estimator $\widehat{\text{FWER}_t}$ as defined in~\eqref{eq:gap_threshold}, and the error control can be proved following Appendix~\ref{apd:thm1}.

\subsection{The gap-railway function}
The proof is the same as that for the gap function except condition~\eqref{eq:cond_gap} is verified for $p$-values with non-decreasing densities differently as follow:
\begin{align*}
    \mathbb{P}(h(P) = 1 \mid g(P) = a)
    =~& \frac{p_l f(a)}{p_l f(a) + (1 - p_u)f(\frac{1 - p_u}{p_l}a + p_u)} {}\\
    =~& \frac{p_l}{p_l + (1 - p_u)f(\frac{1 - p_u}{p_l}a + p_u)/ f(a)} {}\\
    \leq~& \widetilde{p},
\end{align*}
for every $a \in (0, p_l)$.

\section{Varying the parameters in the presented masking functions} \label{apd::para}

We first discuss the original tent masking~\eqref{eq:decomp_varyp}, which represents a class of masking functions parameterized by $p_*$. Similar to the discussion in Section~\ref{sec::mask}, varying $p_*$ also changes the amount of $p$-value information distributed to $g(P)$ for interaction (to exclude possible nulls) and $h(P)$ for error control (by estimating FWER), potentially influencing the test performance. On one hand, the masking function with smaller $p_*$ effectively distributes less information to $g(P)$, in that a larger range of big $p$-values is mapped to small $g(P)$ (see Figure~\ref{fig:tent}). In such a case, the true non-nulls with small $p$-values and small $g(P)$ are less distinctive, making it difficult to exclude the nulls from $\mathcal{R}_t$. On the other hand, the rejected hypotheses in $\mathcal{R}_t^+$ must satisfy $P < p_*$, so smaller $p_*$ leads to less false rejections given the same $\mathcal{R}_t$.

Experiments show little change in power when varying the value of $p_*$ in $(0, \alpha)$ as long as it is not near zero, as it would leave little information in $g(P)$. Our simulations follow the setting in Section~\ref{set:cluster}, where the alternative mean value is fixed at $\mu = 3$. We tried seven values of $p_*$ as $(0.001, 0.005, 0.01, 0.05, 0.1, 0.15, 0.2)$, and the power of the \apt does not change much for $p_* \in (0.05, 0.2)$. This trend also holds when varying the mean value of non-nulls, the size of the grid (with a fixed number of non-nulls), and the number of non-nulls (with a fixed size of the grid). In general, the choice of $p_*$ does not have much influence on the power, and a default choice can be $p_* = \alpha/2$.  

There are also parameters in two other masking functions proposed in Section~\ref{sec::mask}. The railway function flips the tent function without changing the distribution of $p$-value information, hence the effect of varying $p_*$ should be similar to the case in the tent function. The gap function~\eqref{eq:mask_gap} has two parameters: $p_l$ and $p_u$. The tradeoff between information for interaction and error control exhibits in both values of $p_l$ and $p_u$: as $p_l$ decreases (or $p_u$ increases), more $p$-values are available to the analyst from the start, guiding the procedure of shrinking $\mathcal{R}_t$, while the estimation of FWER becomes less accurate. Whether revealing more information for interaction should depend on the problem settings, such as the amount of prior knowledge.

\section{Mixture model for the non-null likelihoods} \label{apd:em}

\paragraph{Two groups model for the $p$-values.} 
Define the $Z$-score for hypothesis $H_i$ as $Z_i = \Phi^{-1}(1 - P_i)$, where $\Phi^{-1}$ is the inverse function of the CDF of a standard Gaussian. Instead of modeling the $p$-values, we choose to model the $Z$-scores since when testing the mean of Gaussian as in~\eqref{eq:normal_hp}, $Z$-scores are distributed as a Gaussian either under the null or the alternative:
\[
H_0: Z_i \overset{d}{=} N(0,1) \quad \text{ versus } \quad H_1: Z_i \overset{d}{=} N(\mu,1),
\]
where $\mu$ is the mean value for all the non-nulls. We model $Z_i$ by a mixture of Gaussians:
\[
Z_i \overset{d}{=} (1 - q_i) N(0,1) + q_i N(\mu,1), \text{ with } q_i \overset{d}{=} \mathrm{Bernoulli}(\pi_i),
\]
where $q_i$ is the indicator of whether the hypothesis $H_i$ is truly non-null.

The non-null structures are imposed by the constraints on $\pi_i$, the probability of being non-null. In our examples, the blocked non-null structure is encoded by fitting $\pi_i$ as a smooth function of the hypothesis position (coordinates)~$x_i$, specifically as a logistic regression model on a spline basis $B(x) = (B_1(x), \ldots, B_m(x))$:
\begin{align} \label{eq:p_model_block}
    \pi_\beta(x_i) = \frac{1}{1 + \exp(- \beta^T B(x_i))},
\end{align}
 
\paragraph{EM framework to estimate the non-null likelihoods.}
\color{black}
An EM algorithm is used to train the model. Specifically we treat the $p$-values as the hidden variables, and the masked $p$-values $g(P)$ as observed. In terms of the $Z$-scores, $Z_i$ is a hidden variable and the observed variable $\Tilde{Z_i}$ is 
\begin{align*}
\Tilde{Z_i} = \begin{cases}
Z_i, & \text{ if }Z_i > \Phi^{-1}(1 - p_*),\\
t(Z_i),  & \text{ otherwise},
\end{cases}
\end{align*}
where $t(Z_i)$ depends on the form of masking. The updates needs values of its inverse function $t^{-1}(\Tilde{Z_i})$ and the derivative of $t^{-1}(\cdot)$, denoted as $\left(t^{-1}\right)'(\Tilde{Z_i})$, whose exact forms are presented below. 
\begin{enumerate}
    \item For tent masking~\eqref{eq:decomp_varyp},
    \begin{align*}
        t(Z_i) =~& \Phi^{-1}\left[1 - \frac{p_*}{1 - p_*}\Phi(Z_i)\right];{}\\
        t^{-1}(\Tilde{Z_i}) =~& \Phi^{-1}\left[\frac{1 - p_*}{p_*}\left(1 - \Phi(\Tilde{Z_i})\right)\right];{}\\
        \left(t^{-1}\right)'(\Tilde{Z_i}) =~& -\frac{1 - p_*}{p_*} \phi\left(\Tilde{Z_i}\right) \Big/ \phi\left(t^{-1}(\Tilde{Z_i})\right),
    \end{align*}
    where $\phi(\cdot)$ is the density function of standard Gaussian.
    
    \item For railway masking~\eqref{eq:mask_railway},
    \begin{align*}
        t(Z_i) =~& \Phi^{-1}\left[1 - p_* + \frac{p_*}{1 - p_*}\Phi(Z_i)\right];{}\\
        t^{-1}(\Tilde{Z_i}) =~& \Phi^{-1}\left[\frac{1 - p_*}{p_*}\left(\Phi(\Tilde{Z_i}) - 1 + p_*\right)\right];{}\\
        \left(t^{-1}\right)'(\Tilde{Z_i}) =~& \frac{1 - p_*}{p_*} \phi\left(\Tilde{Z_i}\right) \Big/ \phi\left(t^{-1}(\Tilde{Z_i})\right).
    \end{align*}
    
    \item For gap masking~\eqref{eq:mask_gap},
    \begin{align*}
        t(Z_i) =~& \Phi^{-1}\left[1 - \frac{p_l}{1 - p^u}\Phi(Z_i)\right];{}\\
        t^{-1}(\Tilde{Z_i}) =~& \Phi^{-1}\left[\frac{1 - p_u}{p_l}\left(1 - \Phi(\Tilde{Z_i})\right)\right];{}\\
        \left(t^{-1}\right)'(\Tilde{Z_i}) =~& -\frac{1 - p_u}{p_l} \phi\left(\Tilde{Z_i}\right) \Big/ \phi\left(t^{-1}(\Tilde{Z_i})\right).
    \end{align*}
    if $Z_i < \Phi^{-1}(1 - p_u)$. If $\Phi^{-1}(1 - p_u) \leq Z_i \leq \Phi^{-1}(1 - p_l)$, which corresponds to the skipped $p$-value between $p_l$ and~$p_u$, then $\Tilde{Z_i} = Z_i$.
    
    \item For gap-railway masking~\eqref{eq:mask_gap_railway},
    \begin{align*}
        t(Z_i) =~& \Phi^{-1}\left[1 - \frac{p_l}{1 - p^u}\Phi(Z_i)\right];{}\\
        t^{-1}(\Tilde{Z_i}) =~& \Phi^{-1}\left[\frac{1 - p_u}{p_l}\left(\Phi(\Tilde{Z_i}) - 1 + p_l\right)\right];{}\\
        \left(t^{-1}\right)'(\Tilde{Z_i}) =~& \frac{1 - p_u}{p_l} \phi\left(\Tilde{Z_i}\right) \Big/ \phi\left(t^{-1}(\Tilde{Z_i})\right).
    \end{align*}
    if $Z_i < \Phi^{-1}(1 - p_u)$. If $\Phi^{-1}(1 - p_u) \leq Z_i \leq \Phi^{-1}(1 - p_l)$, which corresponds to the skipped $p$-value between $p_l$ and~$p_u$, then $\Tilde{Z_i} = Z_i$.
\end{enumerate}

Define two sequences of hypothetical labels $w_i = \one\{Z_i = \Tilde{Z_i}\}$ and $q_i = \one\{H_i = 1\}$, where $H_i = 1$ means hypothesis $i$ is truly non-null ($H_i = 0$ otherwise). The log-likelihood of observing $\Tilde{Z_i}$ is
\begin{align*}
    l(\Tilde{Z_i}) =~& w_i q_i \log\left\{ \pi_i\phi\left(\Tilde{Z_i}-\mu\right)\right\}
    + w_i(1 - q_i) \log \left\{ (1 - \pi_i)\phi\left(\Tilde{Z_i}\right)\right\}{}\\
    +~& (1 - w_i)q_i \log \left\{\pi_i \phi\left(t^{-1}(\Tilde{Z_i})-\mu\right)\right\}
    + (1 - w_i)(1-q_i)\log \left\{(1 - \pi_i) \phi \left( t^{-1}(\Tilde{Z_i})\right)\right\}.
\end{align*}

The E-step updates $w_i,q_i$. Notice that $w_i$ and $q_i$ are not independent, and hence we update the joint distribution of $(w_i, q_i)$, namely
\[
\mathbb{E}[w_iq_i] =: a_i, \quad \mathbb{E}[w_i(1 - q_i)] =: b_i, \quad  \mathbb{E}[(1 - w_i)q_i] =: c_i, \quad \mathbb{E}[(1 - w_i)(1 - q_i)] =: d_i,
\]
where $a_i + b_i + c_i + d_i = 1$. To simplify the expression for updates, we denote
\[
L_i := \pi_i\phi\left(\Tilde{Z_i}-\mu\right) + (1 - \pi_i)\phi\left(\Tilde{Z_i}\right) + \left|\left(t^{-1}\right)'(\Tilde{Z_i})\right| \pi_i \phi\left(t^{-1}(\Tilde{Z_i})-\mu\right) + \left|\left(t^{-1}\right)'(\Tilde{Z_i})\right| (1 - \pi_i) \phi \left( t^{-1}(\Tilde{Z_i})\right).
\]
For the hypothesis $i$ whose $p$-value is masked, the updates are 
\begin{align*}
    a_{i, \text{new}} =~& \mathbb{E}[w_iq_i \mid \Tilde{Z_i}] =  \pi_i\phi\left(\Tilde{Z_i}-\mu\right) \Big/ L_i;{}\\
    b_{i, \text{new}} =~& \mathbb{E}[w_i(1 - q_i) \mid \Tilde{Z_i}] = (1 - \pi_i)\phi\left(\Tilde{Z_i}\right) \Big/ L_i;{}\\
    c_{i, \text{new}} =~& \mathbb{E}[(1 - w_i)q_i \mid \Tilde{Z_i}] = \left|\left(t^{-1}\right)'(\Tilde{Z_i})\right| \pi_i \phi\left(t^{-1}(\Tilde{Z_i})-\mu\right) \Big/ L_i;{}\\
    d_{i, \text{new}} =~& \mathbb{E}[(1 - w_i)(1 - q_i) \mid \Tilde{Z_i}] = \left|\left(t^{-1}\right)'(\Tilde{Z_i})\right| (1 - \pi_i) \phi \left( t^{-1}(\Tilde{Z_i})\right) \Big/ L_i.
\end{align*}
If the $p$-value is unmasked for $i$, the updates are
\begin{align*}
    a_{i, \text{new}} =~& \left(1 + \frac{(1-\pi_i)\phi\left(\Tilde{Z_i}\right)}{ \pi_i\phi\left(\Tilde{Z_i}-\mu\right)}\right)^{-1};{}\\
    b_{i, \text{new}} =~& 1 - a_{i, \text{new}}; \quad c_{i, \text{new}} = 0; \quad d_{i, \text{new}} = 0.
\end{align*}

In the M-step, parameters $\mu$ and $\beta$ (in model~\eqref{eq:p_model_block} for $\pi_i$) are updated. The update for $\mu$ is
\begin{align*}
    \mu_{\text{new}} = \argmax_\mu  \sum_i l(\Tilde{Z_i}) = \frac{\sum a_i\Tilde{Z_i} + c_it^{-1}(\Tilde{Z_i})}{\sum a_i + c_i}.
\end{align*}
The update for $\beta$ is 
\[
\beta_{\text{new}} = \argmax_\beta \sum_i (a_i + c_i)\log \pi_\beta(x_i) + (1-a_i - c_i) \log(1 - \pi_\beta(x_i)),
\]
where $\pi_\beta(x_i)$ is defined in equation~\eqref{eq:p_model_block}. It is equivalent to the solution of GLM (generalized linear model) with the logit link function on data $\{a_i + c_i\}$ using covariates $\{B(x_i)\}$.

\end{document}